\documentclass[11pt,letterpaper,english]{article}

\usepackage{amsmath,amssymb}
\usepackage{amsthm}
\usepackage[ruled,vlined]{algorithm2e}
\usepackage{pgf, tikz}
\usetikzlibrary{arrows}
\usepackage{xspace, units}
\usepackage{typearea}
\allowdisplaybreaks

\usepackage{lmodern}
\usepackage[T1]{fontenc}
\usepackage{textcomp}

\typearea{16}
                             
\usepackage{color}
\definecolor{Darkblue}{rgb}{0,0,0.4}
\definecolor{Brown}{cmyk}{0,0.81,1.,0.60}
\definecolor{Purple}{cmyk}{0.45,0.86,0,0}
\usepackage[breaklinks]{hyperref}
\hypersetup{colorlinks=true,
            citebordercolor={.6 .6 .6},linkbordercolor={.6 .6 .6},
            citecolor=black,urlcolor=Darkblue,linkcolor=black,
            pdfpagemode=UseThumbs,pdfstartview=FitH}
\newcommand{\lref}[2][]{\hyperref[#2]{#1~\ref*{#2}}}

\newtheorem{theorem}{Theorem}
\newtheorem{proposition}[theorem]{Proposition}
\newtheorem{lemma}[theorem]{Lemma}

\newtheorem{claim*}[theorem]{Claim}
\newtheorem{observation}[theorem]{Observation}

\newcommand{\RR}{\ensuremath{\mathbb{R}}}

\newcommand{\e}{{\mathrm e}}

\renewcommand{\Pr}[1]{\mbox{\rm\bf Pr}\left[#1\right]}
\newcommand{\Ex}[1]{\mbox{\rm\bf E}\left[#1\right]}

\newcommand{\OPT}{\mathrm{OPT}}
\DeclareMathOperator*{\ALG}{ALG}

\DeclareMathOperator*{\median}{median}
\DeclareMathOperator*{\arrival}{arrival}
\DeclareMathOperator*{\departure}{departure}

\newcommand{\growingmid}{\mathrel{}\middle|\mathrel{}}

\title{Online Independent Set Beyond the Worst-Case:\\
Secretaries, Prophets, and Periods}

\author{Oliver G\"obel\thanks{Dept.\ of Computer Science, RWTH Aachen University, Germany. {\tt \{goebel,voecking\}@cs.rwth-aachen.de}. Supported by DFG Research Training Group AlgoSyn at RWTH Aachen University.}
\and Martin Hoefer\thanks{Max-Planck-Institut f\"ur Informatik and Saarland University, Saarbr\"ucken, Germany. {\tt mhoefer@mpi-inf.mpg.de}. Supported by DFG Cluster of Excellence M2CI at Saarland University and in part by DFG grant Ho 3831/3-1.}
\and Thomas Kesselheim\thanks{Dept.\ of Computer Science, Cornell University University, Ithaca, NY, USA. {\tt kesselheim@cs.cornell.edu}. Supported by a fellowship within the Postdoc-Programme of the German Academic Exchange Service (DAAD) and by DFG through UMIC Research Center at RWTH Aachen University.}
\and Thomas Schleiden$^*$
\and Berthold V\"ocking$^*$
}

\begin{document}

\maketitle

\begin{abstract}
We investigate online algorithms for maximum (weight) independent set on graph classes with bounded inductive independence number like, e.g., interval and disk graphs with applications to, e.g., task scheduling and spectrum allocation. In the online setting, it is assumed that nodes of an unknown graph arrive one by one over time. An online algorithm has to decide whether an arriving node should be included into the independent set.  Unfortunately, this natural and practically relevant online problem cannot be studied in a meaningful way within a classical competitive analysis as the competitive ratio on worst-case input sequences is lower bounded by $\Omega(n)$. This devastating lower bound holds even for randomized algorithms on unweighted interval graphs and, hence, for one of the most restricted graph class under consideration.

As a worst-case analysis is pointless, we study online independent set in a stochastic analysis. Instead of focussing on a particular stochastic input model, we present a generic sampling approach that enables us to devise online algorithms achieving performance guarantees for a variety of input models. In particular, our analysis covers stochastic input models like the secretary model, in which an adversarial graph is presented in random order, and the prophet-inequality model, in which a randomly generated graph is presented in adversarial order. Our sampling approach bridges thus between stochastic input models of quite different nature. In addition, we show that our approach can be applied to a practically motivated admission control setting in which the algorithm uses the input from a preceding period as sample graph for the current period.

Our sampling approach yields an online algorithm for maximum independent set on interval and disk graphs with competitive ratio $O(1)$ with respect to all of the mentioned stochastic input models.
More generally, for graph classes with inductive independence number~$\rho$, the competitive ratio is $O(\rho^2)$. The approach can be extended towards maximum-weight independent set by losing only a factor of $O(\log n)$ in the competitive ratio with $n$ denoting the (expected) number of nodes.
This upper bound is complemented by a lower bound of $\Omega(\log n/\log^2 \log n)$ showing that our sampling approach achieves nearly the optimal competitive ratio in all of the considered models. Furthermore, we generalize our analysis to address several practically motivated extensions of the independent set problem, e.g., arrival and departure times of nodes or edge-weighted graphs capturing SINR-type interference conflicts in wireless networks. 
\end{abstract}

\thispagestyle{empty}
\setcounter{page}0
\pagebreak[4]

% !TeX root = onlinedisks.tex

\section{Introduction}

Various scheduling and resource allocation problems can be formulated in terms of independent set problems for different graph classes. In such a formulation, the nodes of the graph represent  tasks or requests that are connected by an edge if they are in conflict with each other. An independent set corresponds to a subset of tasks or requests that do not have a conflict and, hence, can be executed or served simultaneously. In the Max-IS problem, the objective is to find an independent set of maximum cardinality. In the Max-Weight-IS problem, the nodes come with weights and the objective is to find an independent set of maximum weight. 

Previous work on independent set problems is mostly concerned with offline optimization where the complete input is known in advance. In many application contexts, however, requests arrive over time and an online algorithm has to make irrevocable decisions for arriving requests without knowing requests that arrive in the future. In particular, in admission control, the online algorithm has to decide which of the arriving requests shall be served and which shall be declined. This corresponds to online variants of independent set where nodes arrive over time. Each node comes with information about its incident edges to previously arrived nodes. The online algorithm has to decide which of the nodes should be included into the independent set and which should be rejected.

Unfortunately, even for rather restrictive (but in the context of scheduling and admission control highly relevant) graph classes like interval and disk graphs, the classical worst-case competitive analysis of online algorithms for independent set problems does not make much sense  -- there is a lower bound of $\Omega(n)$ on the worst-case competitive ratio for Max-IS on $n$-node interval graphs. To see this, consider the following input sequence consisting of $n/2$ pairs of intervals (disks): The sequence begins with a pair of disjoint intervals. Then the next pair of disjoint intervals is inserted into one of these intervals. The decision into which of the intervals from the previous pair the new pair is inserted is done by flipping a fair coin. This process is continued recursively until $n/2$ pairs of intervals have been generated. Obviously, the generated sequence  contains an independent set of $n/2+1$ intervals, regardless of the outcome of the coin flips. A simple recurrence shows, however, that the expected number of independent intervals that can be selected by an online algorithm is at most 2. Thus, the competitive ratio for Max-IS is lower-bounded $\Omega(n)$. By Yao's principle, this negative result extends to randomized online algorithms as well.

The alternative to a worst-case analysis is a stochastic analysis. It is already challenging, however, to choose the {\em right} stochastic input model which, on the one hand, allows for devising online algorithms with meaningful performance guarantees and, on the other hand, is reasonable from a practical point of view.  We approach this challenge by studying not only one but a variety of stochastic input models. In particular, we study two stochastic input models that are inspired by theoretical studies on the secretary problem with arrivals in random order~\cite{Dynkin1963} and a variant of this problem with arrivals in adversarial order combined with stochastic predictions based on so-called prophet-inequalities~\cite{Krengel1977,Krengel1978}. Our study is complemented by a third model that is motivated by a practical admission control problem. In each of these models an input sequence is generated by a different mix of stochastic and adversarial processes.

In our analysis, we focus on graph classes of bounded inductive independence. The inductive independence number $\rho$ of a graph is the smallest number of which there is an order $\prec$ such that for any independent set $S \subseteq V$ and any $v \in V$, we have $\lvert \{ u \in S \mid u \succ v \text{ and } \{u,v\} \in E\} \rvert \leq \rho$. The inductive independence number is a useful concept and bounded in many prominent graph classes. The case of $\rho = 1$ is equivalent to the existence of a so-called perfect elimination ordering, which is used to define the class of chordal graphs \cite{Ye2009}. Classes with larger, but constant $\rho$, include line graphs and planar graphs with $\rho = 2$ and $\rho = 3$, respectively. Line graphs are a special case of claw-free graphs. In these as well as in bounded-treewidth graphs, $\rho$ is bounded in terms of the parameter by which the respective graph class is defined. In case of planar graphs the bound on the inductive independence number is due to the fact that they maintain a bounded average degree on every induced subgraph. In intersection graphs of translates of geometric objects, $\rho$ is constant as well. It depends on the dimension of the geometric object though. In this work, we casually refer to interval and disk graphs, which have bounded inductive independence number $1$ and $5$. The reason for our consideration is, that interval graphs are an established model for scheduling problems in which nodes correspond to tasks with start and finishing times. Disk graphs generalize interval graphs from one to two dimensions. They are frequently used to describe spectrum allocation problems in wireless networks. In addition, we also study the independent set problem with respect to more advanced interference models for wireless networks, whose respective conflict graphs often have bounded inductive independence \cite{Hoefer2011}. This even holds for interference constraints based on SINR (signal to interference plus noise ratio) constraints to which a generalized notion of inductive independence can be applied. 

\subsection{Description of the models}

We study the following stochastic input models:
\begin{itemize}
\item {\em Secretary model:} The adversary defines a node-weighted graph $G=(V,E,w)$ with $n$ nodes. (For simplicity, we assume integer weights. In case of Max-IS, all nodes have weight 1.) A~priori, the algorithm knows $n$ but neither $G$ nor the weights. The nodes of $G$ are presented in random order to the online algorithm where each permutation of the nodes is assumed to occur equally likely.
\item {\em Prophet-Inequality model:} The adversary defines a graph $G=(V,E)$, and for each node a separate probability distribution on its weight. A priori, the algorithm knows $G$ and the probability distributions but not their outcomes. The nodes of $G$ are presented in adversarial order to the online algorithm where the actual weight is revealed to the algorithm only when the node arrives.
\item {\em Period model:} Let $G = (V,E,w)$ be an arbitrary node-weighted graph. Time is partitioned into periods. For each period $t = 1,2,\ldots$ and each node $i \in V$, the adversary defines a probability $p_i^t \in [0,1]$ such that, for $t \ge 2$, $p_i^t \in [p_i^{t-1}/c, p_i^{t-1} \cdot c]$, where $c \ge 1$ is assumed to be constant. Let $X_i^t$ denote independent binary variables with $\Pr{X_i^t=1}=p_i^t$ and $\Pr{X_i^t=0}=1-p_i^t$. Let $V_t = \{i \in V \mid X_i^t = 1\}$. In every period  $t \ge 2$, the nodes in $V_t$ are presented in adversarial order to the online algorithm which aims at finding an independent set among the nodes in $V_t$. In this model, the probabilities $p_i^t$ and the order in which the nodes in $V_t$ arrive are not assumed to be known by the algorithm a priori.
\end{itemize}
The first of these stochastic input models is inspired by the classical secretary problem~\cite{Dynkin1963} in which the task is to pick the best among $n$ secretaries which are presented in random order. The second model is in spirit of problems with prophet-inequality~\cite{Krengel1977,Krengel1978}, where candidates come in adversarial order but each candidate has a publicly known distribution of his weight. In the base case of either setting, one has to select one of $n$ entities that are presented online and have to be accepted or rejected immediately at arrival. Each entity comes with a weight, that is revealed at the time of arrival. The objective is to maximize the weight of the one entity that is accepted. In the secretary problem, the weights are determined by an adversary, but the adversary is not able to fix the order. Instead, entities arrive in a random permutation. In the prophet-inequality model, weights are drawn at random from publicly known distributions, but the adversary can fix 
distributions 
and arrival order.

The third model is motivated by admission control protocols that have to decide about requests using stochastic knowledge from previous ``corresponding'' periods. For example, to make decisions in the time period on this week Friday from 9am to 10am, an admission control algorithm might want to learn from events in the same time window(s) of previous Friday(s). The idea behind this model is that the graph $G$ describes a potentially very large universe of possible requests. This graph might represent disks of various sizes at different positions which are requested with certain probabilities.
An adversary fixes a distribution which generates requests by picking a set of nodes from $G$ at random. 
Distributions might change over time but the deviation from period to period is bounded as specified by the global constant $c \ge 1$. The order in which the requests are presented in the period model is adversarial. That is, it is assumed implicitly that the order in which requests arrive within a period is unpredictable.

We evaluate online algorithms in terms of the {\em competitive ratio} which is defined as $\Ex{\rm OPT}/\Ex{\rm ALG}$ with OPT denoting the maximum weight of an independent set for the given instance and ALG denoting the weight of the independent set selected by the online algorithm. 
The expectation is with respect to the stochastic input model and random coin flips of the algorithm.
In case of the secretary model, the weight of OPT is fixed so that the competitive ratio simplifies to ${\rm OPT}/\Ex{\rm ALG}$. In case of the period model, we study the competitive ratio with respect to any fixed period $t \ge 2$.

\subsection{Our Contribution}
The three stochastic input models described above are conceptually quite different.
In particular, the secretary model assumes an adversarial graph whose nodes are presented in random order, whereas the other two models assume that the graph is randomly generated but the nodes are presented in adversarial order. Furthermore, it is assumed that the online algorithm has access to the distribution in the prophet model, whereas, in the period model, it can only observe samples obtained from similar distributions.

In order to cope with differences between the models, we present a unifying {\em graph sampling model} that bridges between these models. In this model, the online algorithm is initially equipped with a sample graph which is generated by a distribution that is stochastically similar to the distribution of the input graph.
In Section~\ref{sec:model}, we introduce this model formally and show that it can be simulated by each of the other input models. By this approach, we are able to devise online algorithms achieving -- up to small constant factors -- the same competitive ratio for all of the models. 

On the basis of the graph sampling model, we are able to present an online algorithm for Max-IS with competitive ratio $O(\rho^2)$ for graphs with inductive independence number $\rho$. The algorithm and its analysis are presented in Section~\ref{sec:unweighted}. In particular, we achieve competitive ratio $O(1)$ for independent set on interval and disk graphs in all of the considered input models. Our analytic approach shows that one does not need to make specific stochastic assumptions in order to break through the $\Omega(n)$ worst-case lower bound. Indeed, the same kind of online algorithm performs well under a variety of stochastic assumptions.

In Section 4, we present upper and lower bounds for weighted independent set.
At first, we show how the algorithm for Max-IS can be adapted to Max-Weight-IS. We obtain a competitive ratio of $O \left( \rho^2 \log n \right)$ in the graph sampling model and, hence, all of the models, where $n$ denotes an upper bound on the (expected) number of nodes that are presented to the algorithm (in the considered period). We show that this bound is almost best possible for interval and disk graphs. In particular, we prove a lower bound on the competitive ratio for weighted independent set on interval graphs of order $\Omega\left( \log n /\log^2 \log n \right)$ in the secretary and the prophet-inequality model. The same bound applies to the period and the unifying model, too, and it holds even for randomized algorithms.

Motivated by admission control and scheduling applications, we additionally study a problem variant in which nodes have different arrival and departure times. The adversary is allowed to fix in advance the conflict graph and for each node a time interval in which the node is present. Only the nodes being active at the same time have to be independent. Technical details of this online problem are explained in Section~\ref{sec:arrivals}.
We show how to solve variants of Max-IS and Max-Weight-IS with arrival and departure times by using the algorithms from Section 3 and 4 as subroutines.
In particular, by introducing arrival and departure times, we lose only a factor 
\( O\left( \log n \right) \) in the competitive ratio.

Finally, we show how to transfer our results to edge-weighted conflict graphs. This way, more sophisticated wireless interference models can be analyzed, for example, the commonly studied ones based on SINR constraints. We present an algorithm whose competitive ratio is \( O\left( \rho^2 \log^2 n \right) \) for solving online independent set in edge-weighted conflict graphs. A more detailed description of the algorithm and the employed model is given in Section~6.

% !TeX root = onlinedisks.tex

\subsection{Related Work}
\label{sec:related}

The (offline) maximum independent set problem on interval graphs was essentially solved already in the 1970s: Frank~\cite{Frank1975} presented a linear-time algorithm that solves the problem exactly on chordal graphs. For disk graphs, in contrast, the problem is NP-hard but admits a PTAS~\cite{Erlebach2005}. A number of studies have considered graphs of bounded inductive independence number \cite{Akcoglu2000,Ye2009}, which mainly see this property as a generalization of chordal graphs. Graphs of inductive independence number $\rho$ are also $\rho \chi(G)$-inductive. Irani~\cite{Irani1994} shows how to color $d$-inductive graphs online with $O(d \log n)$ colors.

The use of disk graphs was often motivated by interference in wireless networks. In a more realistic model for interference using the signal-to-interference-plus-noise ratio (SINR), a number of approximation algorithms \cite{Andrews2009,Goussevskaia2009,Halldorsson2009,Halldorsson2011,Kesselheim2011} for different variants of maximizing the number of successful simultaneous transmissions have been presented. Interestingly, Hoefer et al.~\cite{Hoefer2011} showed that any of these problems can also be described as a maximum independent set problem in an edge-weighted conflict graph. Moreover, the inductive independence number of these graphs turns out to be bounded by a constant or $O(\log n)$ (see also \cite{Halldorsson2013}). Besides, also the graphs arising from a number of further simple interference models have a constant inductive independence number as well. Therefore, this property is a very useful abstraction when dealing with wireless interference in an algorithmic setting.

In any of these models, online worst-case optimization can only achieve trivial results when not assuming further restrictions. For interval graphs, one commonly assumes the value density (in the unweighted case: ratio of the maximum to the minimum interval length) to be bounded by some $k$. For this case, Koren and Shasha~\cite{Koren1995} present a $(1 + \sqrt{k})^2$-competitive algorithm, which is also shown to be optimal. For SINR models, Fangh\"anel et al.~\cite{OnlineSPAA} use similar geometric parameters and achieve tight competitive ratios. Unfortunately, in comparison to offline optimization the achievable guarantees are quite poor.

One of the earliest results in non-worst-case online optimization is by Dynkin~\cite{Dynkin1963}, who analyzed the basic secretary problem and presented a simple (optimal) selection rule, that accepts the highest ranked entity with probability at least $\frac{1}{\e} - o(1)$. A similar constant-factor approximation is possible in the prophet-inequality setting as shown by Krengel and Sucheston~\cite{Krengel1977,Krengel1978}. Both settings have strong connections to online auctions, where bidders arrive one at a time and have to be served. There is a large body of work on how incentive compatibility can be achieved \cite{Hajiaghayi2007,Chawla2010,Alaei2011}.

Only very recently, these models have been analyzed with respect to combinatorial optimization problems that are non-trivial in the offline setting. The matroid independent set problem was considered in secretary~\cite{Babaioff2007} and prophet-inequality models~\cite{Kleinberg2012}. Different variants of matching and set packing problems were studied in the secretary model \cite{Korula2009,KrystaV12, Kesselheim2013ESA} and the prophet-inequality model~\cite{Alaei2012}. Our algorithm for unweighted independent set is inspired by~\cite{Korula2009}, which uses a greedy algorithm to guide the online-computation of a weighted matching. 

More general packing problems have been studied in the secretary model as well. In its simplest form, this is the knapsack problem, which was considered by Babaioff et al.~\cite{Babaioff2007Knapsack}. Allowing multiple constraints, the problem becomes solving linear packing problems online. It is shown in \cite{Agrawal2009, Devanur2009, Feldman2010, Devanur2011, Molinaro2012} that this kind of online problems can be solved almost optimally, provided that the given capacities are large enough. However, even under this restricting assumption, the described algorithms are not applicable for independent set problems. 

% !TeX root = onlinedisks.tex

\section{Graph Sampling Model}
\label{sec:sampling}

In this section, we present a technically motivated but nevertheless rather intuitive stochastic input model that serves as bridge between the three input models from the introduction. In the {\em graph sampling model}, the online algorithm is initially equipped with a sample graph that is stochastically similar to an input graph which is presented in online fashion. In the following, we first describe the properties of this model formally and then we explain how it can be simulated by each of the other three models so that competitive ratios achieved for the graph sampling model hold for those models, too.

Let $G = (V, E)$ be an arbitrary graph from the considered class.
From this graph, one derives two induced subgraphs, the {\em input graph $G[V^I]$} with weights $w^I$ and the {\em sample graph} $G[V^S]$ with weights $w^S$, where $V^I,V^S \subseteq V$. The two sets $V^I$ and $V^S$ are generated implicitly by drawing non-negative weights $w^I(v)$ and $w^S(v)$ at random, for each node $v \in V$. For simplicity in notation, we assume that node weights are integral.
We set $V^I = \{v \in V| w^I(v) > 0\}$ and $V^S = \{v \in V| w^S(v) > 0\}$. The weights $w^I$ and $w^S$ need not to be drawn according to exactly identical distributions, but they have to satisfy the following assumptions.\footnote{The competitive ratios that we prove do not depend on the size of the graph $G$, but only on the expected size of the graph $G[V^I]$ presented to the online algorithm. For this reason, the model can be extended to infinite graphs representing, e.g.,  all possible disks in Euclidean space. In such an extension, probability distributions might be continuous rather than discrete. Only for notational simplicity, we focus on finite graphs, integer weights, and discrete  probability distributions.}
\begin{itemize}
\item {\em Stochastic similarity:} for every node $v \in V$ and every integer $b > 0$, $\Pr{w^I(v) = b}\leq c\Pr{w^S(v) = b}$ and $\Pr{w^S(v) = b}\leq c\Pr{w^I(v) = b}$
with $c \geq 1$ denoting a fixed, constant term.
\item {\em Stochastic independence:} for every node $v \in V$, the weights $w^I(v)$ and $w^S(v)$ do not depend on the weights $w^I$ and $w^S$ of other nodes
\end{itemize}
Let us explicitly point out that the weights $w^I$ and $w^S$ for the same node need not to be independent, that is, for any $v \in V$, $w^I(v)$ and $w^S(v)$ might be correlated. These possible dependencies are crucial for the simulation of the graph sampling model by the secretary model.

Let us now describe how the input is presented to the online algorithm.
A priori, the algorithm is not supposed to know $G$, the weights $w^I$, $w^S$, or even the probability distributions for the weights. As initial input, it receives the sample graph $G[V^S]$ together with the weights $w^S$ for the nodes in $V^S$. Nodes in $V^I$ arrive one by one in adversarial order. When a node $v \in V^I$ arrives, the algorithm gets to know the weight $w^I(v)$ as well as the edges from $v$ to nodes in $V^S$ and to those nodes in $V^I$ that arrived before $v$. If $v$ is also contained in $V^S$, it is revealed that these are identical. Based on this information, the online algorithm has to irrevocably decide whether $v$ should be included into the independent set or rejected.

The \emph{competitive ratio} of an algorithm in the graph sampling model is defined as $\Ex{{\rm OPT}(w^I)} / \Ex{\rm ALG}$, where ${\rm OPT}(w^I)$ is the maximum weight of an independent set with respect to the weights $w^I$. The next proposition shows that an upper bound on the competitive ratio for the graph sampling model implies upper bounds on the competitive ratios for the other stochastic input models.

\begin{proposition}\label{proposition:sampling1}
If there is an $\alpha$-competitive algorithm for Max-IS (Max-Weight-IS) in the graph sampling model, then there are $O(\alpha)$-competitive algorithms for Max-IS (Max-Weight-IS) in the prophet-inequality model, the period model, and the secretary model.
\end{proposition}

\begin{proof}
We describe how the input and the sample graph and the corresponding weights can be derived in each of the three input models. Let $\alpha = \alpha(c)$ denote the competitive ratio in the graph sampling model with $c$ denoting the constant term from the similarity condition.

At first, we consider the prophet-inequality model. The graph $G=(V,E)$ of the sampling model is identified with the corresponding graph of the prophet-inequality model. We set $V = V^S = V^I$ by assuming, w.l.o.g., that all weights are positive. In particular, the weights $w^I$ in the graph sampling model correspond to the original weights $w$ from the prophet-inequality model. The online algorithm initially generates the weights $w^S$ by simulating the probability distributions given by the prophet-inequality model. Since $w^I$ and $w^S$ are generated by identical distributions, the similarity condition of the sampling model holds for $c=1$. Thus, the competitive ratio in the prophet-inequality model is $\alpha(1)$.

At second, we consider the period model. One obtains a sample graph for period $t \ge 2$ by observing period $t-1$. Let $w(v)$ denote the adversarial weights from the period model. In the graph sampling model, we set $w^S(v) = w(v)$ if $v \in V_{t-1}$ and $w^S(v) = 0$, otherwise. Analogously, we set $w^I(v) = w(v)$ if $v \in V_t$ and $w^I(v) = 0$, otherwise. This way, the similarity condition is satisfied so that the competitive ratio in the period model is $\alpha(c)$, where $c$ corresponds to the parameter in the definition of the probabilities in this model.

Finally, we describe the simulation for the secretary model. We draw a random number $k$ from the Binomial distribution $B(n,\frac12)$ with $n = |V|$ and define $V^S$ to contain the first $k$ nodes of the input and  $V^I$ to contain the remaining $n - k$ nodes. Technically, this is achieved by setting $w^I(v) = 0$, $w^S(v)=w(v)$, for the first $k$ nodes, and $w^I(v) = w(v)$, $w^S(v)=0$, for the remaining nodes, where $w(v)$ denotes the adversarial weights from the secretary model. As nodes arrive in random order and $k$ is determined by the binomial distribution the above definition of the weights is stochastically equivalent to choosing weight tuples $(w^I(v), w^S(v))$ independently, uniformly at random from $\{(0,w(v)),(w(v),0)\}$, for all nodes $v \in V$.  Thus, stochastic independence and stochastic similarity (with $c=1$) are satisfied. 
The online algorithm is $\alpha(1)$-competitive with respect to $G[V^I]$, that is, $\Ex{\rm ALG} \ge  \Ex{{\rm OPT}(w^I)}/\alpha(1)$.
Observe, however, that $G[V^I]$ contains only about half of the nodes of $G$. The other half is used for building the sample graph $G[V^S]$.
By symmetry, $\Ex{{\rm OPT}(w^I)} = \Ex{{\rm OPT}(w^S)}$, which implies ${\rm OPT}(w) = \Ex{{\rm OPT}(w^I + w^S)} \le \Ex{{\rm OPT}(w^I)} + \Ex{{\rm OPT}(w^S)} = 2 \Ex{{\rm OPT}(w^I)}$. Consequently,
$\Ex{\rm ALG} \ge \Ex{{\rm OPT}(w)}/2\alpha(1)$ so that the competitive ratio for the secretary model is upper-bounded by~$2\alpha(1)$.
\end{proof}

Because of Proposition~\ref{proposition:sampling1}, we can focus on the graph sampling model when proving upper bounds on the competitive ratio. The following lemma shows that indeed it is sufficient to compare the independent set computed by the algorithm to the maximum-weight independent set with respect to $w^S$ instead of $w^I$. That is, for the purpose of upper bounding the competitive ratio within constant factors, it suffices to upper-bound $\Ex{{\rm OPT}(w^S)}/\Ex{\rm ALG}$ instead of  $\Ex{{\rm OPT}(w^I)}/\Ex{\rm ALG}$.

\begin{lemma}\label{lemma:sampling2}
$\Ex{{\rm OPT}(w^S)} \ge \frac1{c} \Ex{{\rm OPT}(w^I)}$.
\end{lemma}

\begin{proof}
Stochastic similarity gives $\Pr{w^S(v) = b} \ge \frac{1}{c} \Pr{w^I(v) = b}$, for any $v \in V$ and $b > 0$. W.l.o.g., we can assume that $\Pr{w^S(v) = b} = \frac{1}{c} \Pr{w^I(v) = b}$, for any $b > 0$, as this only decreases $\Ex{{\rm OPT}(w^S)}$. 

Observe that neither $\Ex{{\rm OPT}(w^S)}$ nor $\Ex{{\rm OPT}(w^I)}$ are affected by the correlation between $w^I$ and~$w^S$. This allows us to define an arbitrary coupling between $w^I$ and $w^S$, that is, we rearrange the random experiments for choosing $w^I$ and $w^S$ in a helpful way that changes the correlation between $w^I$ and $w^S$ but does not affect the individual probability distributions for $w^I$ and $w^S$. This coupling is defined as follows: For each node $v \in V$, we set $w^S(v) = w^I(v)$ with probability $\frac1c$ and  $w^S(v) = 0$, otherwise.

Now, following the principle of deferred decisions, we assume that the weights $w^I$ are fixed arbitrarily. Let $U \subseteq V^I$ denote a maximum-weight independent set in $G[V^I]$. Then $U^S = \{v \in U| w^S(v) \ge 0\}$ is an independent set in $G[V^S]$ and $\Ex{w^S(U^S)} = \frac{1}{c} w^I(U) = \frac{1}{c} {\rm OPT}(w^I)$, which implies the lemma.
\end{proof}
\label{sec:model}
% !TeX root = onlinedisks.tex

\section{Unweighted Independent Set}

We study Max-IS on graphs with bounded inductive independence number in the graph sampling model. That is, we consider the input model from Section~\ref{sec:sampling} restricted to $\{0,1\}$-weights and assume that the underlying graph $G=(V,E)$ has bounded inductive independence number $\rho \ge 1$. Because of the restriction to $\{0,1\}$-weights, the graph sampling model simplifies as follows. One picks two subsets $V^I$ and $V^S$ from $V$ at random.
The induced graphs $G[V^I]$ and $G[V^S]$ are the input and the sample graph, respectively. For a node $v \in V$, the events $v \in V^I$ and $v \in V^S$ might be correlated. By the stochastic independence assumption, however, these events, do not depend on events for other nodes. The stochastic similarity property for $\{0,1\}$-weights gives
\begin{equation}
\frac{1}{c}\Pr{v\in V^S}\leq\Pr{v \in V^I}\leq c \, \Pr{v\in V^S}\enspace .
\label{eq:ccondition}
\end{equation}

Our online algorithm applies the greedy algorithm for independent set to the sample graph $G[V^S]$ and employs the output of this algorithm to guide the online computation on the input graph $G[V^I]$. In the offline setting, a greedy algorithm for independent set on graphs with bounded inductive independence number starts with $I=\emptyset$ and considers all nodes of $V$ iteratively according to \(\prec\). It adds a node to \(I\) when it is not in conflict with other nodes already in \(I\). This yields a \(\rho\)-approximation, because the algorithm sticks to the order \(\prec\): Selecting a node not in the optimal solution prevents at most \(\rho\) many neighbors from being selected to \(I\), cf., e.g., \cite{Akcoglu2000, Ye2009}.

In more details, Algorithm~\ref{alg-onlinedisks} computes two sets $M_1 \subseteq V^S$ and $M_2 \subseteq V^I$. $M_1$ is the output of the greedy algorithm applied to $G[V^S]$. $M_2$ is obtained by going through the nodes in $V^I$ in adversarial order and checking for each $v \in V^I$ whether it would have been taken by the greedy algorithm on $G[V^S \cup \{v\}]$. In our analysis, we will show that the expected value of $M_2$ is of the same order as the expected value of $M_1$ and, hence, an $O(\rho)$-approximation of $OPT(w^S)$. By Lemma~\ref{lemma:sampling2}, this implies that $M_2$ is an $O(\rho)$-approximation of $\Ex{OPT(w^I)}$. Unfortunately, however, $M_2$ is not an independent set. Feasibility is achieved by two further steps: We first obtain a set $M_3$ by randomly sparsifying $M_2$, which loses another factor of $O\left(\rho\right)$ in the competitive ratio. The remaining conflicts are resolved by only moving nodes to a set $M_4$ (the output of the algorithm) that are not adjacent to a node that was 
previously inserted into $M_4$. Finally, a stochastic analysis of the conflicts in $M_3$ shows that this final conflict resolution step loses only a constant factor in the competitive ratio.

\begin{algorithm}[t]
%\DontPrintSemicolon
\label{alg-onlinedisks}
\caption{Unweighted Online-Max-IS}
\KwIn{\(G[V^S]\)}
\(M_1, M_2, M_3, M_4 \leftarrow\emptyset\)\;
\ForAll{\(v\in V^S\) in order according to \(\prec\)} 
{\lIf{\(M_1\cup\{v\}\) is independent}{\(M_1\leftarrow M_1\cup\{v\}\)}}
\ForAll{\(v\in V^I\) in order of arrival} {\lIf {\(\nexists u\in M_1,u\prec v\) with \(\{u,v\}\in E\)} {\(M_2\leftarrow M_2\cup\{v\}\)}\;
\lIf{\(v\in M_2\)} {w/prob \( q := \frac{1}{2\rho c} \colon M_3\leftarrow M_3\cup\{v\}\)}\;
\lIf{\(v\in M_3\) and \(\nexists u\in M_4\) s.t. \(\{v,u\}\in E\)}{\(M_4\leftarrow M_4\cup\{v\}\)}}
return \(M_4\)\;
\end{algorithm}
%
%The set \(M_4\), output by Algorithm~\ref{alg-onlinedisks}, clearly is an independent set.
%
\begin{theorem}\label{theorem-onlinedisks}
  Algorithm~\ref{alg-onlinedisks} is $4 c^3 \rho^2$-competitive.
\end{theorem}
To prove the theorem, we will use random variables $X_v^i$ for $i\in\{1,\ldots, 4\}$ where \(X_v^i=1\) if and only if \(v\in M_i\). Let us first observe a fundamental relationship between the respective variables $X_v^1$ and $X_v^2$.

\begin{lemma}
\label{lemma-c-equations}
For any $v \in V$, we have
\begin{equation}\label{eq-c1}
 \frac{1}{c} \Ex{X_v^1} \leq \Ex{X_v^2} \leq c \Ex{X_v^1} 
\end{equation}
and
\begin{equation}\label{eq-c2}
 \frac{1}{c} \Ex{X_v^1 \growingmid u \in M_3} \leq \Ex{X_v^2 \growingmid u \in M_3} \leq c \Ex{X_v^1 \growingmid u \in M_3} 
\end{equation}
for \( u \prec v \). 
\end{lemma}

\begin{proof}
For any $u \in V$, the decision whether $u$ belongs to $M_1 \cup M_2$ or not does not depend on the 
outcome of the random experiment deciding whether a node \( v \succ u \) belongs to $V^S$ or $V^I$. For the analysis, this means that we can assume that the random trial determining the membership in $V^S$ or $V^I$ is actually ``postponed'' as described in Algorithm~\ref{alg:simulate}, which produces stochastically the same sets $M_1$ and $M_2$ as Algorithm~\ref{alg-onlinedisks}. 

\begin{algorithm}
\ForAll{\(v\in V\) in order according to \(\prec\)} 
{\If{\(M_1\cup\{v\}\) is independent}{perform a random experiment with the following outcomes:\\
with probability $\Pr{v \in V^S, v \not\in V^I}$ add $v$ to $M_1$\\
with probability $\Pr{v \not\in V^S, v \in V^I}$ add $v$ to $M_2$\\
with probability $\Pr{v \in V^S, v \in V^I}$ add $v$ to $M_1$ and to $M_2$\\
with probability $\Pr{v \not\in V^S, v \not\in V^I}$ do nothing
}}
\caption{An equivalent way to obtain sets $M_1$ and $M_2$}
\label{alg:simulate}
\end{algorithm}

Let us now consider the course of events that determine if a node $v \in V$ is added to $M_1$, $M_2$, both, or neither by Algorithm~\ref{alg:simulate}. The node can only make it into one of the sets $M_1$ or $M_2$ if there is no $u \prec v$ such that $u \in M_1$ and $\{ u, v \} \in E$. So, let us denote the latter event by $\mathcal{E}$. (Observe that the probability for $\mathcal{E}$ might depend on the condition $u \in M_3$ from Equation \ref{eq-c2}.)

For  for $i \in \{1, 2\}$, $\Ex{X_v^i} = \Pr{ \mathcal{E} } \Pr{X_v^i = 1 \growingmid \mathcal{E} }$. It holds $\Pr{X_v^2 = 1 \growingmid \mathcal{E} } = \Pr{v \in M_2 \growingmid \mathcal{E} } = \Pr{v \in V^I \growingmid \mathcal{E} }$. Now, by stochastic similarity, the latter term is bounded from above by $c \Pr{v \in V^S \growingmid \mathcal{E} } = c \Pr{v \in M_1 \growingmid \mathcal{E} } = c \Pr{X_v^1 = 1 \growingmid \mathcal{E} }$. Observe that the assumption of stochastic independence ensures that the probabilities for $v \in V^S$ and $v \in V^I$ are not influenced by the condition on $\mathcal{E}$. Hence, we have shown the right inequalities of both \eqref{eq-c1} and \eqref{eq-c2}. The left inequalities follow analogously.
\end{proof}

We are now ready to prove bounds on the sets computed by the algorithm. The set \(M_1\) is determined by applying the greedy algorithm for independent set to \(G[V^S]\) and, hence, this set is a \(\rho\)-approximation of \({\rm OPT}(w^S)\). Combining this with Lemma~\ref{lemma:sampling2} showing $\Ex{{\rm OPT}(w^S)} \ge \frac1c \Ex{{\rm OPT}(w^I)}$ gives
\begin{eqnarray}\label{eqn-m1}
\Ex{\left\lvert M_1\right\rvert}\geq \frac{\Ex{{\rm OPT}(w^I)}}{c \rho} \enspace .
\end{eqnarray}
By applying the left inequality from (\ref{eq-c1}), we obtain $\Ex{|M_2|} \ge \frac{1}{c}\Ex{|M_1|}$. Furthermore, sparsifying from $M_2$ to $M_3$ causes losing the factor $q=\frac{1}{2\rho c}$. The reason for choosing $q$ in this way will become clear at the end of the proof. It holds $\Ex{|M_3|} = q \Ex{|M_2|}$. Thus, we obtain
\begin{eqnarray}\label{eqn-m3}
\Ex{\left\lvert M_3\right\rvert} \geq \frac{q \, \Ex{{\rm OPT}(w^I)}}{c^2 \rho} \enspace .
\end{eqnarray}

It remains to analyze the final conflict resolution in Algorithm~\ref{alg-onlinedisks} where only nodes without any conflict are selected in the final output set. The consequence of this approach is that for each conflict which would appear in the offline setting, exactly the node arriving first in the online setting is chosen by our algorithm. For a detailed analysis, we define $C=\{ \{u,v\} \in E \mid u,v \in M_3\}$ as the set of conflicts in \(M_3\). Note that the size of \( C \) is an upper bound to the overall number of nodes that are lost in the conflict resolution.

\begin{lemma}\label{lemma-c}
 \(\Ex{\left\lvert C\right\rvert}\leq\Ex{\left\lvert M_3\right\rvert}q\rho c\).
\end{lemma}
\begin{proof}
We use $C_v=\{u\in V \mid \{u,v\}\in E \text{ and } v\prec u\}$ to denote, for a fixed node \( v \in V \), all nodes larger than \( v \) which are in conflict with \( v \). 
To get a bound on the size of \(C\), we now consider every node \(v \in M_3\) and add up the number of larger nodes which are in conflict with \(v\), enumerated in \(C_v \cap M_3\).  This way, we obtain \( \Ex{ \left \lvert C \right \rvert}=\sum\nolimits_{v\in V} \Pr{v \in M_3} \Ex{ \left \lvert C_v \cap M_3 \right \rvert \growingmid v \in M_3} \).
 
 The size of a set $C_v \cap M_3$ can be expressed in terms of the random variables $X_u^3$ as follows
  \[
  \Ex{ \left \lvert C_v \cap M_3 \right \rvert \growingmid v \in M_3} = \Ex{ \sum_{u \in C_v} X_u^3 \growingmid v \in M_3} = \sum_{u \in C_v} \Ex{X_u^3 \growingmid v \in M_3} \enspace.
  \]
  Next observe that \( \Ex{X_u^3 \growingmid v \in M_3} = q \Ex{X_u^2 \growingmid v \in M_3} \). In combination with the right inequality from~\eqref{eq-c2}, this gives \( \Ex{X_u^3 \growingmid v \in M_3} \le qc \Ex{X_u^1 \growingmid v \in M_3} \), which implies
 \[
  \Ex{ \left \lvert C_v \cap M_3 \right \rvert \growingmid v \in M_3} \leq q c \sum\nolimits_{u \in C_v} \Ex{X_u^1 \growingmid v \in M_3}  \enspace.
  \]
  Finally,  \(M_1\) is independent according to its construction. By definition of the inductive independence number, we have therefore \( \sum_{u \in C_v} X_u^1 \leq \rho \). Combining this with the above bound, we get
  \[
  \Ex{ \left \lvert C_v \cap M_3 \right \rvert \growingmid v \in M_3} \leq q \rho c \enspace.
  \]
 
 The estimation of \(\Ex{\left \lvert C_v \cap M_3 \right \rvert \growingmid v \in M_3}\) combined with the formula for \(\Ex{\left \lvert C \right \rvert}\) from above yields \(\Ex{ \left \lvert C \right \rvert } = \sum_{v \in V} \Pr{v \in M_3} \Ex{ \left \lvert C_v \cap M_3 \right \rvert \growingmid v \in M_3 } \leq \sum_{v \in V} \Pr{v \in M_3}  q \rho c = \Ex{ \left \lvert M_3 \right \rvert} q \rho c\), which proves the lemma. 
\end{proof}

\begin{proof}[Proof of Theorem~\ref{theorem-onlinedisks}]
 The estimated size of the output is exactly the set \(M_3\) from which one node per existing conflict is removed. As consequence, we get \(\Ex{\left\lvert M_4\right\rvert}=\Ex{\left\lvert M_3\right\rvert} - \Ex{\left\lvert C\right\rvert}\) as expectation of the size.
 From the analysis above, we know the estimated sizes of \(M_3\) and \(C\) and get
 \(\Ex{\left\lvert M_4\right\rvert}\geq \Ex{\left\lvert M_3\right\rvert} - q\rho c\Ex{\left\lvert M_3\right\rvert} = \left(1-q\rho c \right) \Ex{\left\lvert M_3\right\rvert} \geq \left(1 - q\rho c\right)\frac{q}{c^2 \rho}\OPT(w)\). Maximizing $\left(1 - q\rho c\right)\frac{q}{c^2 \rho}$ gives $q=\frac{1}{2\rho c}$ as specified in Algorithm~\ref{alg-onlinedisks}. Hence, the competitive ratio is \( 4 c^3 \rho^2 \).
\end{proof}
\label{sec:unweighted}
% !TeX root = onlinedisks.tex

\section{Weighted Independent Set}
\label{sec:weighted}

\subsection{Upper Bound}

In this section, we turn to the Max-Weight-IS problem. We construct an algorithm by dividing nodes into weight classes and running the algorithm for the unweighted problem on a randomly selected class. While this is a common approach in online maximization, we have to deal with the technical difficulties here. Neither $\lvert V^I \rvert$ nor the maximum weight are known a priori and the sample is generated by a stochastically similar rather than by the same distribution. 

The algorithm works as follows: First it ensures that for no node $v$ both $w^I(v)$ and $w^S(v)$ are positive at the same time. Although we lose a factor of 2 in the expected value of the solution this way, this has the advantage that we deal with structurally simpler weight distributions. Afterwards, a random threshold is determined based on the maximum $w^S$ weight. To compute the independent set, each node $v$ whose weight $w^I(v)$ respectively $w^S(v)$ beats this threshold is forwarded to Algorithm~\ref{alg-onlinedisks}. All other nodes are discarded.

\begin{algorithm}
\label{alg-onlineweighteddisks}
\caption{Weighted Online-Max-IS}
for each $v \in V$ flip a fair coin: if heads $w^I(v) := 0$, if tails $w^S(v) := 0$\;
$v^{\max} := \arg\max_{v \in V^S} w^S(v)$\;
$B := w^S(v^{\max})$\;
choose $ X \in \{-1,0,1, \ldots , \left \lceil \log \left( (c+2) \lvert V^S \rvert  \right) \right \rceil \} $ uniformly at random\;
$p := 2^{-X} B $\;
execute Algorithm~\ref{alg-onlinedisks}  with $V^I_p = \{v \in V^I \mid w^I(v) \geq p\}$ and $V^S_p = \{v \in V^S \setminus \{ v^{\max} \} \mid w^S(v) \geq p\}$\;
\end{algorithm}

\begin{theorem}
 \label{theorem-onlineweighteddisks}
  Algorithm~\ref{alg-onlineweighteddisks} is $O(\alpha \cdot \log(n))$-competitive, where $\alpha = O(\rho^2)$ is the competitive ratio of Algorithm~\ref{alg-onlinedisks} and $n = \Ex{ \lvert V^I \vert }$.
\end{theorem}

\begin{proof}
After executing the first line of Algorithm~\ref{alg-onlineweighteddisks}, it holds that at most one of the two weights of a node is positive. For the remainder of the proof, we assume w.l.o.g. that already the original weights $w^I$ and $w^S$ satisfy this condition. Observe that this modification of the weights decreases the expected value of the optimal solution by at most two and preserves the stochastic similarity and independence conditions of the graph sampling model. Let $\hat{w}(v) = \max\{w^I(v), w^S(v)\}$. It holds $(w^I(v), w^S(v)) = (\hat{w}(v),0)$ or $(w^I(v), w^S(v)) = (0,\hat{w}(v))$.

In the following, we make a case distinction depending on properties of the largest and second largest weights. Observe that assumptions on the size of the largest weight influence the probability space of all weights and lead to vast dependencies. In order to cope with this problem, we work on the following conditional probability space: For each node $v \in V$, we assume that  $\hat{w}(v)$ is fixed arbitrarily. Besides, we remove those nodes from our consideration for which $\hat{w}(v)=0$. This way, $V = V^I \,\dot\cup\, V^S$ with
$V^I=\{v \in V| w^I(v) > 0\}$ and $V^S = \{ v \in V | w^S(v) > 0\}$. Stochastic similarity gives 
\[
\Pr{w^S(v) = \hat{w}(v)} \ge \frac{1}{c}  \Pr{w^I(v) = \hat{w}(v)} 
\enspace ,
\]
which, by $V = V^I \,\dot\cup\, V^S$, implies
\begin{equation}\label{eqnWeighted}
\Pr{v \in V^S} \ge \frac{1}{c+1} \enspace ,
\end{equation}
for each node $v \in V$, independent of the outcome of weights of other nodes.

Now let $S^\ast \subseteq V$ be an independent set maximizing $\sum_{v \in S^\ast} \hat{w}(v)$ and let $S$ be the random variable denoting the independent set that our algorithm computes. Let $v^\ast \in \arg\max_{v \in V} \hat{w}(v)$ and $B^\ast := \hat{w}(v^\ast)$. Further $v ' \in \arg\max_{v \in V\backslash \{ v^\ast \}} \hat{w}(v)$ and $B'=\hat{w}(v')$, i.e. $v^\ast$ and $v'$ are the nodes with the highest respectively second highest weight with respect to $\hat{w}(v)$.

\paragraph{Case 1: $B^\ast > \frac{1}{4} \hat{w}(S^\ast)$}

We condition on the event that $v' \in V^S$, which by Equation~\eqref{eqnWeighted} holds with probability at least $\frac{1}{c + 1}$.

If $B^\ast \geq 2 B'$, we further condition on $X = -1$. In this case, there are only two possible inputs for Algorithm~\ref{alg-onlinedisks}: Either $V^I_p = \{ v^\ast \}$ and $V^S_p = \emptyset$ or $V^I_p = V^S_p = \emptyset$. As this does not change the output, we may as well assume the input in the latter case to be $V^I_p = \emptyset$ and $V^S_p = \{ v^\ast \}$. In this case, the unweighted algorithm returns an independent set of expected value at least $\frac{B^\ast}{\alpha (c+1)}$. Hence, we get $\Ex{w^I(S) \growingmid v' \in V^S} \geq \frac{1}{\lceil \log((c+2) \lvert V^S \rvert) \rceil + 1} \cdot \frac{B^\ast}{\alpha (c+1)}$. 

If $B^\ast < 2 B'$, we condition instead on $X = 0$. The input for Algorithm~\ref{alg-onlinedisks} can now be multiple nodes. However, each of them has value at least $\frac{B^\ast}{2}$. Since -- as long as there is a node in $V^I$ -- the maximum-cardinality independent set among these nodes has (trivially) size at least $1$, the unweighted algorithm returns an independent set of value at least $\frac{B^\ast}{2 \alpha (c+1)}$. So, we now get $\Ex{w^I(S) \growingmid v' \in V^S} \geq \frac{1}{\lceil \log((c+2) \lvert V^S \rvert) \rceil + 1} \cdot \frac{B^\ast}{2 \alpha (c+1)}$.

In either case, multiplying with $\Pr{v' \in V^S} \geq \frac{1}{c + 1}$ and using that $B^\ast > \frac{1}{4} \hat{w}(S^\ast)$ yields $\Ex{w^I(S)} = \Omega\left(\frac{1}{\log \lvert V^S \rvert} \right) \Ex{\OPT(w^I)}$. As $\lvert V^S \lvert = O(\lvert V^I \rvert)$ with high probability, this yields the result.

\paragraph{Case 2: $B^\ast \leq \frac{1}{4} \hat{w}(S^\ast)$}

In this case, we condition on the event that $v^\ast \in V^S$. By Equation \ref{eqnWeighted} this event holds with probability at least $\frac{1}{c + 1}$.

For $j = 0,\ldots \lceil \log \lvert V \rvert \rceil+1$, let now $S^\ast_j = \{ v \in S^\ast \mid v \neq v^\ast, \hat{w}(v) \geq 2^{-j} B^\ast \}$. Let us observe that these sets make a significant part of the value of $S^\ast$. Firstly, after removing all nodes $v$ with $\hat{w}(v) < \frac{1}{2 \lvert V \rvert}$, the value of $S^\ast$ can have decreased by at most $\lvert V \rvert \frac{B^\ast}{2 \lvert V \rvert} \leq \frac{\hat{w}(S^\ast)}{2}$. Also, $v^\ast$ by assumption has value at most $\frac{1}{4} \hat{w}(S^\ast)$. That is, after removing both kinds of nodes from $S^\ast$, it has still value at least $\frac{1}{4} \hat{w}(S^\ast)$. For this reason, we get $\sum_{j=0}^{\lceil \log \lvert V \rvert \rceil + 1} 2^{-j} B^\ast \lvert S^\ast_j \rvert \geq \frac{1}{2} \cdot \frac{1}{4} \hat{w}(S^\ast)$.

On the other hand, if $X = j$ and $v^\ast \in V^S$, the algorithm returns a solution of expected size at least $\frac{\lvert S^\ast_j \rvert}{\alpha}$. 

Let now $L = \max\{\lceil \log |V| \rceil, \left \lceil \log \left( (c+2) \lvert V^S \rvert  \right) \right \rceil \}$. For the time being, let us assume that the algorithm chooses the value of $X$ from $\{-1,0,1, \ldots , L\}$ instead of $\{-1,0,1, \ldots , \left \lceil \log \left( (c+2) \lvert V^S \rvert  \right) \right \rceil \}$. Then, we obtain
\[\Ex{w^I(S) \growingmid v^\ast \in V^S} \geq \frac{1}{L + 2} \sum_{j=-1}^{\lceil \log \lvert V \rvert \rceil + 1} \frac{\lvert S^\ast_j \rvert}{\alpha} \geq  \frac{1}{L + 2} \cdot \frac{1}{8 \alpha} \cdot \hat{w}(S^\ast) \enspace .
\]
Finally, observe that the assumption of $X$ being chosen from $\{-1,0,1, \ldots , L\}$ fails only if
$(c+2) \lvert V^S \rvert< \lvert V \rvert$. This event, however, occurs only with probability at most $\frac{1}{\lvert V \rvert}$. Let us assume pessimistically that whenever this unlikely event occurs the algorithm outputs a set of weight 0 whereas the modified algorithm returns a set of weight $\hat{w}(S^\ast)$. Even this decreases the expected value by an additive term of at most $\hat{w}(S^\ast)/ \lvert V \rvert$. This completes the analysis for the second case and, hence, proves the theorem.
\end{proof}

\subsection{Lower Bound}
The approach taken in Algorithm~\ref{alg-onlineweighteddisks} seems fairly generic. However, the competitive ratio turns out to be almost optimal not only for the general problem but even for all special cases mentioned in the introduction.
\begin{theorem}\label{theorem:lower-bound}
For any algorithm for online maximum-weight independent set, we have $\Ex{\ALG} \linebreak= \Omega\left( \frac{\log^2 \log n}{\log n} \right) \Ex{\OPT}$, even in interval graphs, and even in the secretary and prophet-inequality model.
\end{theorem}

\begin{proof}
We will use a fixed graph, which is known to the algorithm in advance. The node weights are drawn independently from probability distributions that are also known in advance. The precise outcome, however, is only revealed at time of arrival. The times of arrival are in uniform random order. This way, we restrict the adversary to become weaker than in both the secretary and the prophet-inequality models.

We set 
\[
d = \left\lceil \max\left\{ 4 \left(2 \frac{\log n}{\log \log n} + 2\right)^2 + 1, \sqrt{\log n} \right\} \right\rceil
\]
and construct the graph by nesting intervals into each other, starting with an interval of length $1$ and continuing by always putting $d$ intervals of length $d^{-i}$ next to each other into an interval of length $d^{-i+1}$. As we reach $n$, all levels except for the last one are complete. This way, we get a graph that is a complete $d$-ary tree with the only difference that there are ``shortcuts'' on the paths from the root to the leaves skipping over some levels.

For convenience, we remove the last level if it is incomplete. This way, we get a graph having at most $n$ nodes. For the number of levels $h$ it has, we get
\[
\left\lceil \frac{\log n}{\log d} \right\rceil - 1 \leq h \leq \left\lceil \frac{\log n}{\log d} \right\rceil + 1 \enspace,
\]
and in particular $h = \Theta\left( \frac{\log n}{\log \log n} \right)$.

For a node $v$ on level $i$, we set the weight at random by
\[
w(v) = \begin{cases}
d^{h-i} & \text{with probability $p$} \\
0 & \text{with probability $1-p$}
\end{cases}
\]
where $p = \frac{1}{2 h}$.

From now on, we will be focussing on the paths from the root node to the leaves that include $h$ nodes each. There are in total $d^h$ such paths. In every independent set, there can be at most one node on any path. In case a node has non-zero weight, its weight directly corresponds to the number of paths it lies on. Therefore, we can equivalently express the weight of an independent set by the number of paths that are covered, i.e., on which a non-zero node is selected.

To show the lower bound, we show that on this graph no online algorithm can be better than the following \textsc{HighStakes} policy. This algorithm accepts a node if and only if it has non-zero weight and there is no more ancestor to come that could cover this node. In other words, we reject a node of non-zero weight if there is a chance that an ancestor could still be selected (because it has not arrived and none of its other descendants have been selected so far). Otherwise, we accept it.

\begin{lemma}
\textsc{HighStakes} is optimal on this graph.
\end{lemma}

\begin{proof}
We show optimality by induction. Suppose that in a graph of $n$ nodes, we have seen all but $n'$ nodes in the random order. We now claim that for any $n'$ it is optimal to continue using the \textsc{HighStakes} policy. For $n' = 0$ this is trivial because we do not have any further choice. So let us turn to the case $n' + 1$. In this case, we are presented a node and have to decide whether or not to select it. If the node has zero weight, it is obviously optimal to reject it. So, let us assume that this node $v$ is located in level $i$ and has weight $w(v) = d^{h-i}$. We now distinguish between the cases that this node can still be covered by an ancestor or not.

The node cannot be covered by an ancestor if every ancestor has already occurred in the random order or cannot be selected anymore because some other descendant has already been selected. In this case, it is optimal to select the node because it covers all paths it lies on and does not prevent nodes on different paths from being added.

If the node can still be covered by an ancestor, we claim that it is better to not select the node. Roughly speaking, this is due to the fact that there are at least $(d-1) w(v)$ paths that cannot be covered by this very ancestor anymore (which will have non-zero weight with probability $p$) while we only gain a weight of $w(v)$ in the solution. For a formal proof, we use the induction hypothesis that subsequently it is optimal to use the \textsc{HighStakes} policy.

\begin{claim*}
Given a path of $k$ vertices that have not been shown or excluded up to now, let $P(k)$ be the probability that this path is covered by \textsc{HighStakes}. Then
\[
P(k) \geq \frac{p}{2k} + P(k-1) \enspace.
\]
\end{claim*}

\begin{proof}
The path is only covered in the next step, if the very first vertex shows up and has non-zero weight. In any other case (a different vertex shows up or this vertex has zero weight), we keep on waiting. In these cases, we are confronted with a path of $k - 1$ vertices. Therefore, we get the following recursion for $P(k)$ when $k > 0$
\[
P(k) = \frac{1}{k} \left(p + (1-p) P(k-1)\right) + \frac{k-1}{k} P(k-1) = \frac{p}{k} + \left( 1 - \frac{p}{k} \right) P(k-1) \enspace,
\]
and $P(0) = 0$.

To bound the growth of this recursion, we use that $P(k - 1)$ is bounded by the probability that there is any non-zero vertex on the path at all, which is given by $1 - ( 1 - p )^{k-1}$. For $k \leq h$, by the definition of $p$, we have $\left(1 - p\right)^{k - 1} \leq \left(1 - \frac{1}{2h} \right)^h \geq \frac{1}{2}$ and therefore $P(k - 1) \leq \frac{1}{2}$. Putting this into the recursion, we get for $k \leq h$ that
\[
P(k) \geq \frac{p}{2k} + P(k-1) \enspace.
\]
\end{proof}

Hence, when accepting this current vertex now, we reduce the probability of at least $(d-1) b$ paths to be covered by at least $\frac{p}{2h}$. We have $h \leq 2 \frac{\log n}{\log \log n} + 2$ and $d \geq 4 \left(2 \frac{\log n}{\log \log n} + 2\right)^2 + 1$. That is $\frac{2h}{p} = 4 h^2 \leq d - 1$, or equivalently 
\[
\frac{p}{2h} (d - 1) w(v) \geq w(v) \enspace.
\]
The overall expected value of the solution is smaller when accepting the currently considered node in comparison to rejecting it. Therefore, it is better to reject it.
\end{proof}

Having shown that \textsc{HighStakes} is optimal, we only need to compare its expected value with the one of a feasible offline solution. \textsc{HighStakes} accepts the $j$th node on a path only if it has non-zero weight and if the $j-1$ nodes on higher levels occur before this node in the random order. The combined probability of this event is $\frac{p}{j}$. Therefore, the overall probability that any node on a path of length $h$ is accepted is at most $\sum_{j=1}^h \frac{p}{j} = O(p \log h)$. The expected value of the solution computed by \textsc{HighStakes} is exactly the expected number of paths that are covered. By the above considerations, we get
\[
\Ex{\ALG} \leq d^h O(p \log h) = O \left( d^h \frac{1}{2h} \log \log n\right) = O \left( d^h \frac{\log^2 \log n}{\log n} \right) \enspace.
\]
On the other hand, we get a feasible offline solution by greedily accepting vertices going down the tree. In this procedure a path of length $h$ is only left uncovered if all nodes on it have zero weight. This happens with probability $(1 - p)^{h} = (1 - \frac{1}{2h})^h \leq \frac{1}{\sqrt{e}}$. Therefore, this solution has value
\[
\Ex{\OPT} \geq d^h \left( 1 - \frac{1}{\sqrt{e}} \right) = \Omega(d^h) \enspace.
\]
In total we have $\Ex{\ALG} = \Omega\left( \frac{\log^2 \log n}{\log n} \right) \Ex{\OPT}$, showing the claim.
\end{proof}

\section{Arrivals and Departures}
Interval graphs are often motivated by problems in which two tasks cannot be processed at the same time. Disk graphs in turn capture the requirement of spatial separation. In this section, we introduce an approach to combine both temporal and spatial separation. Again, we assume that requests are nodes in a graph $G = (V, E)$, which models the geometric properties. Furthermore, each node $v \in V$ in this graph has an arrival time $\arrival(v) \in \RR$ and a departure time $\departure(v)$. We say that $u \in V$ and $v \in V$ are conflicting if $\{ u, v \} \in E$ and $[\arrival(u), \departure(u)] \cap [\arrival(v), \departure(v)] \neq \emptyset$. For a set of nodes $S$, we define $S[a, d] = \{v \in S \mid \arrival(v) \geq a, \departure(v) \leq d \}$ and $S[a, x, d] = \{v \in S \mid a \leq \arrival(v) \leq x \leq \departure(v) \leq d \}$. Still, we make no assumption on the order in which requests in a period are presented to the online algorithm. In particular, this includes the most natural case, in which 
requests are ordered by arrival times.

In the following, we present an algorithm for this combined problem. We assume to be given an algorithm $\mathcal{A}$ that approximately solves the online (weighted) independent set problem on the graph $G$ with competitive ratio $\gamma$. Using this algorithm, we achieve $O(\gamma \log n)$ as the overall competitive ratio, where $ n = \lvert V^I \rvert $. In the previous sections, we devised such algorithms for the unweighted independent set problem with $\gamma = O(\rho^2)$ and the weighted variant with $\gamma = O(\rho^2 \log n)$. To introduce arrival and departure times, we use a recursive approach that is inspired by a divide-and-conquer algorithm for maximum-weight independent set in rectangle graphs by Agarwal et al.~\cite{Agarwal1998}.

Algorithm \textsc{Split} works as follows. It uses the set $V^S$ as a guide to split the set $V^I$ into three parts. For this purpose, it determines the median $x_{\text{med}}$ of all arrival times in $V^S$ and splits $V^I$ into the set of requests that are active at time $x_{\text{med}}$ and the ones that depart before $x_{\text{med}}$ respectively arrive after $x_{\text{med}}$. Since all requests of the first kind are active simultaneously, we can treat this subproblem by using the given algorithm. The remaining, inactive requests, in contrast, can be treated as two independent, smaller instances. On these we can apply the algorithm recursively. Our algorithm decides randomly, whether to solve the active requests or to recursively invoke and merge the two currently inactive instances.

Formally, the algorithm sets $\tilde{n} = (c+2) \lvert V^S \rvert$, $\chi = 1 - \frac{1}{2(c+2)}$, and $\delta = \frac{24 (c+1)^2 \ln \tilde{n}}{\chi}$. Afterwards, it executes the recursive procedure \textsc{Split}~($- \infty$, $\infty$, $\tilde{n}$).

\begin{algorithm}
\DontPrintSemicolon
\If{$k \geq \delta$}
{Let $x_{\text{med}} = \median_{v \in V^S[a,d]} \arrival(v)$\;
w/prob $q = \frac{2 \gamma}{\alpha \log(k) + \beta}$ run $\mathcal{A}$ on all $v$ such that $a \leq \arrival(v) \leq x_{\text{med}} < \departure(v) \leq d$\;
otherwise run \textsc{Split}($a$, $x_{\text{med}}$, $\chi k$) and \textsc{Split}~($x_{\text{med}}$, $d$, $\chi k$) and join the output\;
}
\Else{
Choose each node $v \in V^I[a, d]$ with probability $\frac{1}{\delta}$.
} 
\caption{\textsc{Split}($a$, $d$, $k$)}
\end{algorithm}

\begin{theorem}
\label{theorem:arrivals}
The algorithm is $O(\gamma \log n)$-competitive.
\end{theorem}

To show the bound, we will first show that in every execution of \textsc{Split}, the number of requests dealt with in $V^I$ is significantly reduced -- even though the set $V^S$ is used for splitting.

Our proof builds on the fact that it is very unlikely that there are much more $V^S$ requests that $V^I$ requests or vice versa in an interval that is large enough. Formally, we call the sets $V^S$ and $V^I$ \emph{balanced} if for any $a \leq d$ with $\lvert (V^S \cup V^I)[a, d]\rvert \geq 24 (c+1)^2 \ln n$, we have $\frac{1}{c+1} \lvert V^S[a, d] \rvert \leq \lvert V^I[a, d] \rvert \leq (c+1) \lvert V^S[a, d] \rvert$.

Similar to the proof of Theorem~\ref{theorem-onlineweighteddisks}, we assume that without loss of generality that for each $v \in V$, we have $\Pr{v \not\in V^S, v \not\in V^I} = 0$ as stochastic similarity and independence are preserved in this conditioned probability space. This way, the set $V^S \cup V^I$ is not a random variable anymore. As $\lvert V^S \cup V^I \rvert = O(\lvert V^I \rvert)$ with high probability, we redefine $n = \lvert V^S \cup V^I \rvert$.

\begin{lemma}
\label{lemma:balanced}
The sets $V^S$ and $V^I$ are balanced with probability at least $1 - \frac{1}{n^2}$.
\end{lemma}

We show that for any $a \leq d$ with $\lvert (V^S \cup V^I)[a, d]\rvert \geq 24 (c+1)^2 \ln n$, we have 
\[
\Pr{\lvert V^I[a, d] \rvert > (c+1) \lvert V^S[a, d] \rvert} \leq \frac{1}{n^4} \enspace.
\]
This shows the claim because there are at most $\frac{n(n-1)}{2}$ different values for $a$ and $d$ that need to be considered and, furthermore, the bound on $\Pr{\lvert V^S[a, d] \rvert > (c+1) \lvert V^I[a, d] \rvert}$. The rest follows then by applying a union bound.

After simplifying notation, it suffices to show the following claim.

\begin{claim*}\label{claim:5-random-variables}
Let $A_1, \ldots, A_k, B_1, \ldots, B_k$ be $0$/$1$ random variables with the following properties
\begin{itemize}
  \item $A_1 (1 - B_1), \ldots, A_k (1 - B_k)$ are independent
  \item For all $i \in [k]$, we have $\Pr{A_i = 0, B_i = 0} = 0$, $\Pr{A_i = 1} \leq c \cdot \Pr{B_i = 1}$ for some $c \geq 1$.
\end{itemize}
Then we have
\[
\Pr{\sum_{i \in [k]} A_i \geq (c + 1) \sum_{i \in [k]} B_i} \leq \exp\left( - \frac{k}{6 (c+1)^2}\right)
\]
\end{claim*}
\begin{proof}
Let $X_i = 1$ iff $A_i = 1$ and $B_i = 0$; let $Y_i = 1$ iff $A_i =1$ and $B_i = 1$. By this definition, we have $A_i = X_i + Y_i$ and $B_i = 1 - X_i$. Furthermore $\sum_{i \in [k]} A_i \geq (c + 1) \sum_{i \in [k]} B_i$ is equivalent to
\[
\sum_{i \in [k]} (X_i + Y_i) \geq (c + 1) \sum_{i \in [k]} (1 - X_i) = (c + 1) k - (c + 1) \sum_{i \in [k]} X_i \enspace.
\]
This is equivalent to
\[
(c + 2) \sum_{i \in [k]} X_i + \sum_{i \in [k]} Y_i \geq (c + 1) k \enspace.
\]
For the probability, this means
\begin{align*}
\Pr{\sum_{i \in [k]} A_i \geq (c + 1) \sum_{i \in [k]} B_i} & = \Pr{(c + 2) \sum_{i \in [k]} X_i + \sum_{i \in [k]} Y_i \geq (c + 1) k} \\
& \leq \Pr{\sum_{i \in [k]} X_i \geq \frac{c}{c + 2} k} \enspace.
\end{align*}
As $\Pr{A_i = 1} \leq c \cdot \Pr{B_i = 1}$, we have $\Pr{X_i = 1} = \Pr{A_i = 1, B_i = 0} \leq \Pr{A_i = 1} \leq \frac{1}{c} \Pr{B_i = 1} =  \frac{1}{c} (1 - \Pr{B_i = 0}) \leq \frac{1}{c} (1 - \Pr{A_i = 1, B_i = 0}) = \frac{1}{c} (1 - \Pr{X_i = 1})$. Therefore, we need to have $\Pr{X_i = 1} \leq \frac{c}{c+1}$. Using a Chernoff bound,
\begin{align*}
\Pr{\sum_{i \in [k]} X_i \geq \left( \frac{c}{c+1} \right)^2 k} & = \Pr{\sum_{i \in [k]} X_i \geq \left( 1 - \frac{1}{c+1} \right) \left( \frac{c}{c+1} \right) k} \\
& \leq \exp\left( - \frac{\left(\frac{1}{c+1} \right)^2 \left( \frac{c}{c+1} \right) k}{3}\right) \enspace.
\end{align*}
As $\left(\frac{c}{c+1}\right)^2 \leq \frac{c}{c+1} \leq \frac{c}{c+2}$ and $\left( \frac{1}{c+1} \right)^2 \left( \frac{c}{c+1} \right) \geq \frac{1}{2 (c+1)^2}$ for $c \geq 1$, this yields
\[
\Pr{\sum_{i \in [k]} X_i \geq \frac{c}{c+2} k} \leq \exp\left( - \frac{k}{6 (c+1)^2}\right) \enspace,
\]
which shows the claim.
\end{proof}

Having established this property, the algorithm behaves nicely in the sense that $x_{\text{med}}$ is always a reasonable point to split.

\begin{claim*}\label{claim:5-balanced}
If $V^S$ and $V^I$ are balanced, for any execution of \textsc{Split}($a$, $d$, $k$), we have $\lvert (V^S \cup V^I)[a, d]\rvert \leq k$.
\end{claim*}

\begin{proof}
We prove this claim by induction. For the very first execution \textsc{Split}~($a$, $d$, $\tilde{n}$), the claim holds because if $V^S$ and $V^I$ are balanced, then $\lvert V^S \cup V^I \rvert \leq (c+2) \lvert V^S \rvert = \tilde{n}$.

Any other execution of \textsc{Split}~($a$, $d$, $k$) is a recursive call by some other execution \textsc{Split}~($a'$, $d'$, $k'$) for some $k' \geq \delta$, where $k = \chi k'$. W.l.o.g., consider the execution on the ``left'' subinstance, in which $a$ is set to $a'$ and $d$ is set to $x_{\text{med}}$, the median of all arrival times on $V^S[a', d']$. Observe that if $\lvert (V^S \cup V^I)[a', d'] \rvert <\chi \delta \leq \chi k' = k$, we are immediately done. So, we can assume that $\lvert (V^S \cup V^I)[a', d'] \rvert \geq \chi \delta$. Directly from the definition, we get
\[
\lvert (V^S \cup V^I)[a', x_{\text{med}}] \rvert \quad \leq \quad \lvert (V^S \cup V^I)[a', d'] \rvert - \lvert V^S[a', d'] \setminus V^S[a', x_{\text{med}}] \rvert \enspace.
\]

As $\lvert (V^S \cup V^I)[a', d'] \rvert \geq \chi \delta$ and $V^S$ and $V^I$ are balanced, we have that $\lvert V^S[a', d'] \rvert \geq \frac{1}{c+2} \lvert (V^S \cup V^I) [a', d'] \rvert$. Combining this fact with the median property implies
\[
\lvert V^S[a', d'] \setminus V^S[a', x_{\text{med}}] \rvert  \quad \geq \quad \frac{1}{2} \lvert V^S[a', d'] \rvert \quad \geq \quad \frac{1}{2(c+2)} \lvert (V^S \cup V^I)[a, d] \rvert \enspace.
\]
In combination, $\lvert (V^S \cup V^I)[a', x_{\text{med}}] \rvert \leq \chi \lvert (V^S \cup V^I)[a, d] \rvert \leq \chi k' = k$.
\end{proof}

Using these two technical claims, we can now proceed to the proof of Theorem~\ref{theorem:arrivals}. In fact, we show the following stronger claim which immediately yields the desired result.

\begin{claim*}\label{claim:5-split} For any fixed $a < d$ and $k$ such that $\lvert (V^S \cup V^I)[a,d] \rvert \leq k$, we have
\[
\Ex{ \textsc{Split}(a, d, k) \growingmid \mathcal{E} } \quad \geq \quad \frac{1}{\alpha \log(k) + \beta} \cdot \Ex{\OPT(w^I)[a,d]}\enspace,
\]
where $\alpha = \frac{2 \gamma}{- \log \left( \chi \right)}$, and $\beta = 2 \e \delta$.
\end{claim*}

\begin{proof}
We show this claim by induction on $k$.

The base case is $k < \delta$. In this case, we randomly choose a request to accept. For each request, the probability of being accepted is at least $\left( 1 - \frac{1}{\delta} \right)^{\delta - 1} \frac{1}{\delta} \geq \frac{1}{\e \delta}$. This yields
\[
\Ex{\textsc{Split}(a, d, k) \growingmid \mathcal{E}} \quad \geq \quad \frac{1}{\e \delta} \cdot \Ex{ \OPT(w^I)[a, d] \growingmid \mathcal{E} } \enspace.
\]
Furthermore, we have $\Ex{\OPT(w^I)[a, d] \growingmid \mathcal{E}} \geq \frac{1}{2(c+1)} \OPT(w^I)[a, d]$. This yields 
\begin{align*}
\Ex{\textsc{Split}(a, d, k) \growingmid \mathcal{E}} \quad &\geq \quad \frac{1}{\e \delta} \cdot \Ex{\OPT(w^I)[a, d] \growingmid \mathcal{E}}\\
&\geq \quad \frac{1}{2 \e \delta} \cdot \Ex{\OPT(w^I)[a, d]}\\
&\geq \quad \frac{1}{\alpha \log(k) + \beta} \cdot \Ex{\OPT(w^I)[a, d]}\enspace.
\end{align*}

For the induction step, let us consider the case that $k \geq \delta$. By induction hypothesis,
\[
\Ex{ \textsc{Split}(a, x, \chi k) \growingmid \mathcal{E} } \geq  \frac{1}{\alpha \log\left(\chi k \right) + \beta} \cdot \Ex{\OPT(w^I)[a, x_{\text{med}}]}
\]
and
\[
\Ex{ \textsc{Split}(x, d, \chi k) \growingmid \mathcal{E} } \geq  \frac{1}{\alpha \log\left(\chi k\right) + \beta} \cdot \Ex{\OPT(w^I)[x_{\text{med}}, d]} \enspace.
\]
Furthermore, by assumption on the algorithm,
\[
\Ex{ \ALG(a, x, d) \growingmid \mathcal{E}} \quad \geq \quad \frac{1}{2} \Ex{ \ALG(a, x, d) } \quad \geq \quad \frac{\Ex{\OPT(w^I)[a, x_{\text{med}}, d]}}{2 \gamma} \enspace.
\]
In combination, we get for the expectation of the value
\begin{align*}
& \Ex{\textsc{Split}(a, d, k) \growingmid \mathcal{E}} \\
& \geq \frac{2 \gamma}{\alpha \log(k) + \beta} \cdot \frac{\Ex{\OPT(w^I)[a, x_{\text{med}}, d]}}{2 \gamma} \\
& \quad + \frac{\alpha \log\left(k\right) + \beta - 2 \gamma}{\alpha \log\left(k\right) + \beta} \cdot \frac{1}{\alpha \log\left(\chi k\right) + \beta} \cdot \left( \Ex{\OPT(w^I)[a, x_{\text{med}}]} + \Ex{\OPT(w^I)[x_{\text{med}}, d]} \right) \enspace.
\end{align*}
By definition of $\alpha$, we have $\alpha \log\chi \leq - 2\gamma$. Therefore, we get
\begin{align*}
& \Ex{\textsc{Split}(a, d, k) \growingmid \mathcal{E} } \\
& \geq \frac{1}{\alpha \log(k) + \beta} \cdot \Ex{ \OPT(w^I)[a, x_{\text{med}}, d] + \OPT(w^I)[a, x_{\text{med}}] + \OPT(w^I)[x_{\text{med}}, d]) \growingmid \mathcal{E} } \enspace.
\end{align*}
This shows the claim because for any $x_{\text{med}} \in [a, d]$ and any weights $w$ we have
\[
\OPT(w^I)[a, d] \quad \leq \quad \OPT(w^I)[a, x_{\text{med}}, d] + \OPT(w^I)[a, x_{\text{med}}] + \OPT(w^I)[x_{\text{med}}, d] \enspace.
\]
\end{proof}\label{sec:arrivals}
\section{Edge-Weighted Conflict Graphs}
To capture more realistic wireless interference models such as for example the ones based on the signal-to-interference-plus-noise ratio (SINR), we now extend our study to edge-weighted conflict graphs, following the approach in \cite{Hoefer2011}. We assume that between any pair of nodes $u, v \in V$, there exists a (directed) weight $w(u, v) \in [0, 1]$.  It will be convenient to use undirected weights $\bar{w}(u, v) = w(u, v) + w(v, u)$ and setting $w(v, v)=0$ for all $v \in V$.

We define $S \subseteq V$ to be an independent set if $\sum_{u \in S} w(u, v) < 1$ for all $v \in S$. The inductive independence number is now the smallest number $\rho$ for which there is an ordering $\prec$ such that for all independent sets $S$, we have $\sum_{u \in S, u \succ v} \bar{w}(u, v) \leq \rho$ for all $v \in V$.

The major challenge and main distinction over the case of unweighted conflict graphs is that conflicts become asymmetric. In unweighted graphs, node $u$ has a conflict with node $v$ if and only if node $v$ has a conflict with node $u$. In edge-weighted conflict graphs, there might be many nodes $u_1,u_2,\ldots$ that can feasibly be placed into the independent set when considering previously added nodes, but this might violate some other node $v$ that was added before. Our solution to this approach is an additional thinning step in the construction of set $M_3$, which allows to build the final set $M_4$ and lose only an additional polylogarithmic factor in the competitive ratio.

\begin{algorithm}
\label{alg-onlineedgeweighted}
\DontPrintSemicolon
$\tilde{n} := (c+2) \lvert V^S \rvert$\;
\ForAll{$v \in V^S$ in order of increasing $\prec$}{
\lIf{$\sum_{u \in M_1} \bar{w}(u, v) < 1$}{
add $v$ to $M_1$}
}
\ForAll{$v \in V^I$ in order of arrival}{
\lIf{$\sum_{u \in M_1, u \prec v} \bar{w}(u, v) < 1$}{add $v$ to $M_2$}\;
\lIf{$v \in M_2$}{w/prob $q$ add $v$ to $M_3$}\;
\lIf{$v \in M_3$ and there is no $u \in M_3$ with $\bar{w}(u, v) \geq \frac{1}{4 \e \log \tilde{n}}$}{add $u$ to $M_4$}
}
\caption{Online Max-IS in Edge-Weighted Graphs}
\end{algorithm}

Similar to our approach in the last section, we assume w.l.o.g.\ that for each $v \in V$ we have $\Pr{v \not\in V^S, v \not\in V^I} = 0$ and redefine $n = \lvert V^S \cup V^I \rvert$.

Note that the set $M_1$ is feasible with respect to only the preceding nodes in the $\prec$-ordering. However, given such a set, the contention resolution procedure in~\cite{Hoefer2011} can be used to decompose this set into $O(\log n)$ many feasible independent sets. 

Hence, the following observation is immediate.
\begin{observation}
The set $M_1$ can be decomposed to $2 \lceil \log n \rceil$ feasible solutions.
\end{observation}
The set $M_4$, which is the output set of the algorithm, is not ensured to be feasible either. However, we will show that with high probability it is feasible. Therefore, one can ensure feasibility without loss by adding an arbitrary conflict-resolution filter at the end.

Similar to our analysis of Algorithm~\ref{alg-onlinedisks}, our main concern will again be bounding the effects of the conflict resolution that is performed when deriving the set $M_4$ from $M_3$. For this purpose, we first obtain a high-probability bound on the sums of edge weights within set $M_2$.

\begin{lemma}\label{lemma:6-probability}
Let $\tau(n) := \frac{3(1+c)^3 + c}{c}\left( 2 \rho \lceil \log n \rceil + 1 \right)$. With probability at least $1 - \frac{1}{n}$, we have $\sum_{u \in M_2} \bar{w}(u, v) \leq \tau(n)$ for all $v \in M_2$.
\end{lemma}

\begin{proof}
Similar to our analysis in the previous section, w.l.o.g., we assume that for any $v \in V$ the probability that $v$ is neither in $V^S$ nor in $V^I$ is $0$.
 
Let us consider an arbitrary $v \in V$. We show that only with probability at most $\frac{1}{n^2}$, both $v \in M_2$ and $\sum_{u \in M_2} \bar{w}(u, v) > \tau(n)$. If $v \in M_2$, then we have $\sum_{u \in M_1, u \succ v} \bar{w}(u, v) \leq 2 \rho \lceil \log n \rceil$ and $\sum_{u \in M_1, u \prec v} \bar{w}(u, v) < 1$. That is, for $v \in M_2$, the necessary condition $\sum_{u \in M_1} \bar{w}(u, v) \leq 2 \rho \lceil \log n \rceil + 1$ has to be fulfilled. In the following, we fix the outcome whether $v \in V^S$ and/or $v \in V^I$.
 
Let $V \setminus \{v\} = \{u_1, \ldots, u_{n-1}\}$, where $u_1 \prec u_2 \prec \ldots \prec u_{n-1}$. Let $X_i = 1$ if $u_i \in V^I$ but $u_i \not\in V^S$, $0$ otherwise. Note that for each node $u_i$, the fact whether $u_i \in M_1 \cup M_2$ is determined by the random variables $X_1$, \ldots, $X_{i-1}$. The decision whether $u_i$ is added to $M_1$ or $M_2$ is then determined by $X_i$. Furthermore, as for each $u_i$ the probabilities of being contained in $V^S$ or in $V^I$ may differ by at most a factor of $c$, we get $\Ex{X_i} \leq \frac{c}{1+c}$.
 
For each $i$, we define a weight as follows. The weight $w_i$ is a function of $X_1, \ldots, X_{i-1}$: We set $w_i(X_1, \ldots, X_{i-1})$ to $\bar{w}(u_i, v)$ if the outcomes $X_1$, \ldots, $X_{i-1}$ yield $u_i \in M_1 \cup M_2$. Otherwise, we set $w_i(X_1, \ldots, X_{i-1})$ to $0$. This way, the equations $\sum_{u \in M_2 \setminus M_1} \bar{w}(u, v) = \sum_{i=1}^{n-1} w_i(X_1, \ldots, X_{i-1})  X_i$ and $\sum_{u \in M_1} \bar{w}(u, v) =$ $\sum_{i=1}^{n-1} w_i(X_1, \ldots, X_{i-1})  (1 - X_i)$ hold. As the random variables $X_1, \ldots, X_{n-1}$ are independent, we can apply Lemma~\ref{lemma:piggybankprobability}, whose proof can be found below.
 
\begin{lemma}
\label{lemma:piggybankprobability}
Let $X_1$, $X_2$, \ldots, $X_n$ independent Bernoulli trials such that $\Ex{X_i} \leq r$ for all $i \in [n]$. For each $i \in [n]$, let $w_i\colon \{0, 1\}^{i-1} \to [0, 1]$ be an arbitrary function. Let $B \geq 1$. Then we have
\[
\Pr{\sum_{i=1}^n w_i(X_1, \ldots, X_{i-1}) X_i \geq \left( \frac{3}{(1-r)^2 r} \right) B, \sum_{i=1}^n w_i(X_1, \ldots, X_{i-1}) (1 - X_i) \leq B} \leq \exp(-B) \enspace .\]
\end{lemma}
 
Setting $r = \frac{c}{1+c}$ and $B = 2 \rho \lceil \log n \rceil + 1$, we get
\begin{align*}
& \Pr{\sum_{u \in M_1 \cup M_2} \bar{w}(u, v) > \frac{3(1+c)^3 + c}{c}\left( 2 \rho \lceil \log n \rceil + 1 \right), \sum_{u \in M_1} \bar{w}(u, v) \leq 2 \rho \lceil \log n \rceil + 1} \\
& \leq \Pr{\sum_{u \in M_2 \setminus M_1} \bar{w}(u, v) > \frac{3}{\left( 1 - \frac{c}{1+c} \right)^2 \frac{c}{1+c} }\left( 2 \rho \lceil \log n \rceil + 1 \right), \sum_{u \in M_1} \bar{w}(u, v) \leq 2 \rho \lceil \log n \rceil + 1} \\
& \leq \exp\left(- 2 \rho \lceil \log n \rceil + 1 \right) \leq \frac{1}{n^2} \enspace.
\end{align*}
Applying a union bound now yields the claim.
\end{proof}

The low edge weight within $M_2$ is the key step to show that the additional thinning step towards $M_3$ and $M_4$ allows to obtain a feasible independent set with high probability.

\begin{theorem}
\label{theorem:feasibility}
For $q \leq \frac{1}{2\e \tau(\tilde{n})}$, the set $M_4$ is a feasible independent set with probability at least $1 - \frac{3}{n}$.
\end{theorem}

\begin{proof}
Let $\mathcal{E}$ be the event that $n \leq \tilde{n}$ and $\sum_{u \in M_2} \bar{w}(u, v) \leq \tau(n)$ for all $v \in M_2$. So far, we have shown that $\Pr{\mathcal{E}} \geq 1 - \frac{2}{n}$. We show that, conditioned on the event $\mathcal{E}$, $M_4$ is feasible with probability at least $1 - \frac{1}{n}$. For this purpose, we define reduced weights by setting $w'(u, v) = \min\{ \frac{1}{4 \e \log n}, \bar{w}(u, v) \}$ for each pair of nodes $u, v \in V$. Observe that it suffices to have feasibility in $M_4$ with respect to the reduced weights $w'$ because any pair of nodes connected by an edge of weight at least $\frac{1}{4 \e \log \tilde{n}} \leq \frac{1}{4 \e \log n}$ is deleted from $M_3$ to get $M_4$.

For each $v \in M_2$, we have $\Ex{ \sum_{u \in M_3} w'(u, v) \growingmid \mathcal{E}} \leq \Ex{ \sum_{u \in M_3} \bar{w}(u, v) \growingmid \mathcal{E}} \leq \tau(n) q \leq \frac{1}{2 \e}$. The random variable $\sum_{u \in M_3} w'(u, v)$ is a weighted sum of independent $0$/$1$ random variables in which all weights are at most $\frac{1}{4 \e \log n}$. Therefore, we can apply a Chernoff bound to get
\[
\Pr{ \sum_{u \in M_3} w'(u, v) \geq 1 \growingmid \mathcal{E} } \leq \Pr{ \sum_{u \in M_3} w'(u, v) \geq 2 \e \cdot \frac{1}{2 \e} \growingmid \mathcal{E} } \leq 2^{- \frac{1}{2 \e} \big/ \frac{1}{4 \e \log n}} = \frac{1}{n^2} \enspace.
\]
A union bound yields
\[
\Pr{\text{$M_4$ feasible} \growingmid \mathcal{E}} \geq 1 - \frac{1}{n} \enspace,
\]
and, hence, $\Pr{\text{$M_4$ feasible}} \geq 1 - \frac{3}{n}$.
\end{proof}

Finally, we are ready to show the result on the approximation guarantee of our algorithm.

\begin{theorem}\label{theorem:edgeweighted-apx}
If we pick $q \leq \frac{1}{(4 \e \log \tilde{n}) \tau(\tilde{n})}$, we have $\Ex{\lvert M_4 \rvert} \geq \frac{q}{4 \rho} \Ex{\OPT(w^S)}$, i.e., our algorithm is $O(\rho^2 \log^2 n)$-competitive.
\end{theorem}

\begin{proof}
Due to the greedy algorithm, we have $\lvert M_1 \rvert \geq \frac{1}{\rho} \OPT(w^S)$ for 0/1 weights $w^S$. Furthermore, for the same reasons as in the proof of Theorem~\ref{theorem-onlinedisks}, we have again $\Ex{ \lvert M_2 \rvert} \geq \frac{1}{c} \Ex{\lvert M_1 \rvert}$. 

Again, let $\mathcal{E}$ be the event that $n \leq \tilde{n}$ and $\sum_{u \in M_2} \bar{w}(u, v) \leq \tau(n)$ for all $v \in M_2$. For each node $v$ let $X_v = \left\lvert \left\{u \in M_3 \growingmid \bar{w}(u, v) \geq \frac{1}{4 \e \log n} \right\} \right\rvert$. Observe that $\Ex{ X_v \growingmid v \in M_3, \mathcal{E} } \leq \frac{1}{2}$ if $q \leq \frac{1}{(4 \e \log n) \tau(n)}$.

The event $v \in M_3$ but $v \not\in M_4$ only occurs when $X_v \geq 1$. By the above considerations, we have $\Pr{v \not\in M_4 \growingmid v \in M_3, \mathcal{E}} = \Pr{X_v \geq 1 \growingmid v \in M_3, \mathcal{E}} \leq \frac{1}{2}$. That is, $\Pr{ v \in M_4 \growingmid v \in M_3, \mathcal{E}} \geq \frac{1}{2}$, and therefore $\Ex{\lvert M_4 \rvert \growingmid M_2, \mathcal{E}} \geq \frac{q}{2} \lvert M_2 \rvert$.

In combination
\[
\Ex{\lvert M_4 \rvert} \; \geq \; \Pr{\mathcal{E}} \cdot \frac{q}{2} \cdot \Ex{\lvert M_2 \rvert \growingmid \mathcal{E}} \; \geq \; \Pr{\mathcal{E}} \cdot \frac{q}{2} \cdot \frac{1}{\rho} \cdot \OPT(w^S) \; \geq \; \frac{q}{4 \rho} \cdot \OPT(w^S) \enspace,
\]
proving the theorem.
\end{proof}

\subsection{Proof of Lemma~\ref{lemma:piggybankprobability}}
To show Lemma~\ref{lemma:piggybankprobability}, we note that, without loss of generality, we can assume that for all possible $0$/$1$ vectors $(x_1, \ldots, x_n) \in \{0, 1\}^n$, we have $\sum_{i=1}^n w_i(x_1, \ldots, x_{i-1}) = (\frac{3}{(1-r)^2 r} + 1) B$. Otherwise, we can increase the probability of the event to occur by modifying the respective functions. Under these circumstances, we can apply the following variant of the Chernoff bound.

\begin{lemma}
Let $X_1$, $X_2$, \ldots, $X_n$ independent Bernoulli trials such that $\Ex{X_i} \leq r$ for all $i \in [n]$. For each $i \in [n]$, let $w_i\colon \{0, 1\}^{i-1} \to [0, 1]$ be a function such that for all $(x_1, \ldots, x_n) \in \{0, 1\}^n$ we have $\sum_{i=1}^n w_i(x_1, \ldots, x_{i-1}) = \frac{\mu}{r}$. Then we have for all $\delta > 0$
\[
\Pr{\sum_{i = 1}^n w_i(X_1, \ldots, X_{i-1}) X_i \geq (1 + \delta) \mu} \leq \left( \frac{\e^\delta}{(1 + \delta)^{1+\delta}} \right)^\mu \enspace.
\]
\end{lemma}

\begin{proof}
For all $t > 0$, we have
\begin{align*}
& \Pr{\sum_{i = 1}^n w_i(X_1, \ldots, X_{i-1}) X_i \geq (1 + \delta) \mu} \\
& = \Pr{\exp\left( t \sum_{i = 1}^n w_i(X_1, \ldots, X_{i-1}) X_i \right) \geq \e^{t(1 + \delta) \mu}} \\
& = \e^{-t(1 + \delta) \mu} \cdot \Ex{\exp\left( t \sum_{i = 1}^n w_i(X_1, \ldots, X_{i-1}) X_i \right)}
\end{align*}
For $t = \ln(1+\delta)$, we show that $\Ex{\exp\left( t \sum_{i = 1}^n w_i(X_1, \ldots, X_{i-1}) X_i \right)} \leq \exp(\delta \mu)$ by induction.

\begin{align*}
& \Ex{\exp\left( t \sum_{i = 1}^n w_i(X_1, \ldots, X_{i-1}) X_i \right)} \\
& = \Pr{X_1 = 1} \Ex{\exp\left( t w_1 + \sum_{i = 2}^n w_i(1, X_2, \ldots, X_{i-1}) X_i \right)} \\
& \qquad + \Pr{X_1 = 0} \Ex{\exp\left( \sum_{i = 2}^n w_i(0, X_2, \ldots, X_{i-1}) X_i \right)} \\
& = \Pr{X_1 = 1} \e^{ t w_1 } \Ex{\exp\left( \sum_{i = 2}^n w_i(1, X_2, \ldots, X_{i-1}) X_i \right)} \\
& \qquad + \Pr{X_1 = 0} \Ex{\exp\left( \sum_{i = 2}^n w_i(0, X_2, \ldots, X_{i-1}) X_i \right)} \enspace.
\end{align*}
By induction hypothesis $\Ex{\exp\left( \sum_{i = 2}^n w_i(1, X_2, \ldots, X_{i-1}) X_i \right)} \leq \exp\left( \delta(\mu - w_1 r) \right)$ and also \linebreak $\Ex{\exp\left( \sum_{i = 2}^n w_i(0, X_2, \ldots, X_{i-1}) X_i \right)} \leq \exp\left( \delta(\mu - w_1 r) \right)$. Plugging this in, we get
\[
\Ex{\exp\left( t \sum_{i = 1}^n w_i(X_1, \ldots, X_{i-1}) X_i \right)} \leq \left( \Pr{X_1 = 1} \e^{ t w_1 } + \Pr{X_1 = 0} \right)\exp\left( \delta(\mu - w_1 r) \right) \enspace.
\]
As $\Pr{X_1 = 1} \leq 1$ and $r \e^{ t w_1 } + (1 - r) \leq \exp(r w_1 \delta)$, we get $\Ex{\exp\left( t \sum_{i = 1}^n w_i(X_1, \ldots, X_{i-1}) X_i \right)} \leq \e^{\delta \mu}$.
\end{proof}

To show the lemma, we set $\mu = r (\frac{3}{(1-r)^2 r} + 1) B$, $\delta = 1-r$. With these definitions, we have
\[
(1 + \delta) \mu \leq \left( \frac{3}{(1-r)^2 r} \right) B
\]
and therefore
\begin{align*}
\Pr{\sum_{i=1}^n w_i(X_1, \ldots, X_{i-1}) X_i \geq \left( \frac{3}{(1-r)^2 r} \right) B} & \leq \Pr{\sum_{i=1}^n w_i(X_1, \ldots, X_{i-1}) X_i \geq (1 - \delta) \mu} \\
& \leq \exp\left( -\frac{\delta^2 \mu}{3} \right) \\
& = \exp\left( - \frac{(1-r)^2 r (\frac{3}{(1-r)^2 r} + 1) B}{3} \right) \\
& \leq \exp(- B) \enspace.
\end{align*}
\label{sec:edgeweighted}
\section{Conclusion}
In this paper, we present a novel approach to evaluate online algorithms for packing problems. We have concentrated on the maximum independent set problem in graphs of bounded inductive independence because a large number of practically relevant problems can be captured this way. Besides, in the offline setting pretty simple algorithms already achieve sensible approximation guarantees. In the worst-case online setting, however, even the unweighted case one can only achieve trivial guarantees. 

Our approach covers and generalizes a variety of stochastic input models, including recently popular ones derived from secretary problems and prophet inequalities. For our most general graph sampling model, we design online algorithms that allow to obtain near-optimal competitive ratios. It is likely that such a unifying model can also be applied in other domains. For example, a number of existing algorithms for the secretary model \cite{Babaioff2007,Korula2009,Chakraborty2012} can be naturally generalized in a similar way.

Stochastic input models have recently found prominent applications in the context of online auctions, where the goal is to sell items to bidders in order to maximize social welfare or seller revenue. In this case, bidders arrive online one by one, and their benefits for the items are private information that have to be elicited in a truthful way. Instead, we take a more general and fundamental approach that is not tailored to specific auction scenarios. Nevertheless, some of our results can be applied in the context of auctions as well. For example, our algorithm in Section~\ref{sec:weighted} can be turned into a truthful online auction, when its weight is the private value of node for being in the independent set. It is possible to combine our allocation algorithm with appropriate payments (see~\cite{Hoefer2013Universal}) such that revealing the weight truthfully becomes a dominant strategy for each node $v$.

There are many interesting open problems for further research stemming from our paper. Some immediate problems concern, e.g., the simplification of our algorithms or optimization of ratios and parameters. For example, our algorithms, especially the ones for edge-weighted conflict graphs, apply a thinning step to construct $M_3$, which drastically reduces the probability that any conflict occurs. It would be interesting to see if such a thinning step can be omitted by using a more clever analysis of the resulting graph structure. In addition, our algorithms are designed to tackle the general case of the graph sampling model, thereby sacrificing constants in the competitive ratio. It might be possible to design algorithms with better ratios for more constrained input models (like secretary or prophet-inequality models). Finally and more fundamentally, a central requirement in our analysis is that distributions for different nodes are independent. In fact, this assumption is also critical in existing approaches 
to ordinary secretary and prophet-inequality settings, and it would be interesting to see how this assumption can be relaxed. 

\bibliography{online}

\begin{thebibliography}{10}

\bibitem{Agarwal1998}
Pankaj~K. Agarwal, Marc~J. van Kreveld, and Subhash Suri.
\newblock Label placement by maximum independent set in rectangles.
\newblock {\em Comput. Geom.}, 11(3-4):209--218, 1998.

\bibitem{Agrawal2009}
Shipra Agrawal, Zizhuo Wang, and Yinyu Ye.
\newblock A dynamic near-optimal algorithm for online linear programming.
\newblock {\em CoRR}, abs/0911.2974, 2009.

\bibitem{Akcoglu2000}
Karhan Akcoglu, James Aspnes, Bhaskar DasGupta, and Ming-Yang Kao.
\newblock Opportunity cost algorithms for combinatorial auctions.
\newblock {\em CoRR}, cs.CE/0010031, 2000.

\bibitem{Alaei2011}
Saeed Alaei.
\newblock Bayesian combinatorial auctions: Expanding single buyer mechanisms to
  many buyers.
\newblock In {\em Proceedings of the 52nd IEEE Annual Symposium on Foundations
  of Computer Science (FOCS)}, pages 512--521, 2011.

\bibitem{Alaei2012}
Saeed Alaei, MohammadTaghi Hajiaghayi, and Vahid Liaghat.
\newblock Online prophet-inequality matching with applications to ad
  allocation.
\newblock In {\em Proceedings of the 13th ACM Conference on Electronic Commerce
  (EC)}, pages 18--35, 2012.

\bibitem{Andrews2009}
Matthew Andrews and Michael Dinitz.
\newblock Maximizing capacity in arbitrary wireless networks in the {SINR}
  model: Complexity and game theory.
\newblock In {\em Proceedings of the 28th Conference of the IEEE Communications
  Society (INFOCOM)}, pages 1332--1340, 2009.

\bibitem{Babaioff2007Knapsack}
Moshe Babaioff, Nicole Immorlica, David Kempe, and Robert Kleinberg.
\newblock A knapsack secretary problem with applications.
\newblock In {\em Proceedings of the 10th International Workshop on
  Approximation Algorithms for Combinatorial Optimization Problems
  (APPROX-RANDOM)}, pages 16--28, 2007.

\bibitem{Babaioff2007}
Moshe Babaioff, Nicole Immorlica, and Robert Kleinberg.
\newblock Matroids, secretary problems, and online mechanisms.
\newblock In {\em Proceedings of the 18th ACM-SIAM Symposium on Discrete
  Algorithms (SODA)}, pages 434--443, 2007.

\bibitem{Chakraborty2012}
Sourav Chakraborty and Oded Lachish.
\newblock Improved competitive ratio for the matroid secretary problem.
\newblock In {\em Proceedings of the 23rd ACM-SIAM Symposium on Discrete
  Algorithms (SODA)}, pages 1702--1712, 2012.

\bibitem{Chawla2010}
Shuchi Chawla, Jason~D. Hartline, David~L. Malec, and Balasubramanian Sivan.
\newblock Multi-parameter mechanism design and sequential posted pricing.
\newblock In {\em Proceedings of the 42nd annual ACM symposium on Theory of
  computing (STOC)}, pages 311--320, 2010.

\bibitem{Devanur2009}
Nikhil~R. Devanur and Thomas~P. Hayes.
\newblock The adwords problem: online keyword matching with budgeted bidders
  under random permutations.
\newblock In {\em Proceedings of the 10th ACM Conference on Electronic Commerce
  (EC)}, pages 71--78, 2009.

\bibitem{Devanur2011}
Nikhil~R. Devanur, Kamal Jain, Balasubramanian Sivan, and Christopher~A.
  Wilkens.
\newblock Near optimal online algorithms and fast approximation algorithms for
  resource allocation problems.
\newblock In {\em Proceedings of the 12th ACM Conference on Electronic Commerce
  (EC)}, pages 29--38, 2011.

\bibitem{Dynkin1963}
Eugene~B Dynkin.
\newblock The optimum choice of the instant for stopping a markov process.
\newblock In {\em Sov. Math. Dokl}, volume~4, pages 627--629, 1963.

\bibitem{Erlebach2005}
Thomas Erlebach, Klaus Jansen, and Eike Seidel.
\newblock Polynomial-time approximation schemes for geometric intersection
  graphs.
\newblock {\em SIAM J. Comput.}, 34(6):1302--1323, 2005.

\bibitem{OnlineSPAA}
Alexander Fangh{\"a}nel, Sascha Geulen, Martin Hoefer, and Berthold
  V{\"o}cking.
\newblock Online capacity maximization in wireless networks.
\newblock In {\em Proceedings of the 22nd ACM Symposium on Parallelism in
  Algorithms and Architectures (SPAA)}, pages 92--99, 2010.

\bibitem{Feldman2010}
Jon Feldman, Monika Henzinger, Nitish Korula, Vahab~S. Mirrokni, and Clifford
  Stein.
\newblock Online stochastic packing applied to display ad allocation.
\newblock In {\em Proceedings of the 18th annual European Symposium on
  Algorithms (ESA)}, pages 182--194, 2010.

\bibitem{Frank1975}
Andr\'as Frank.
\newblock Some polynomial algorithms for certain graphs and hypergraphs.
\newblock In {\em Proc. 5th British Combinatorial Conference}, pages 211--226,
  1975.

\bibitem{Goussevskaia2009}
Olga Goussevskaia, Roger Wattenhofer, Magn{\'u}s~M. Halld{\'o}rsson, and Emo
  Welzl.
\newblock Capacity of arbitrary wireless networks.
\newblock In {\em Proceedings of the 28th Conference of the IEEE Communications
  Society (INFOCOM)}, pages 1872--1880, 2009.

\bibitem{Hajiaghayi2007}
Mohammad~Taghi Hajiaghayi, Robert~D. Kleinberg, and Tuomas Sandholm.
\newblock Automated online mechanism design and prophet inequalities.
\newblock In {\em Proceedings of the 22nd Conference on Artificial Intelligence
  (AAAI)}, pages 58--65, 2007.

\bibitem{Halldorsson2013}
Magn{\'u}s~M. Halld{\'o}rsson, Stephan Holzer, Pradipta Mitra, and Roger
  Wattenhofer.
\newblock The power of non-uniform wireless power.
\newblock In {\em Proceedings of the 24th ACM-SIAM Symposium on Discrete
  Algorithms (SODA)}, pages 1595--1606, 2013.

\bibitem{Halldorsson2011}
Magn{\'u}s~M. Halld{\'o}rsson and Pradipta Mitra.
\newblock Wireless capacity with oblivious power in general metrics.
\newblock In {\em Proceedings of the 22nd ACM-SIAM Symposium on Discrete
  Algorithms (SODA)}, pages 1538--1548, 2011.

\bibitem{Halldorsson2009}
Magn\'{u}s~M. Halld\'{o}rsson and Roger Wattenhofer.
\newblock Computing wireless capacity.
\newblock unpublished manuscript, 2010.

\bibitem{Hoefer2013Universal}
Martin Hoefer and Thomas Kesselheim.
\newblock Universally truthful secondary spectrum auctions.
\newblock {\em CoRR}, abs/1305.2350, 2013.

\bibitem{Hoefer2011}
Martin Hoefer, Thomas Kesselheim, and Berthold V{\"o}cking.
\newblock Approximation algorithms for secondary spectrum auctions.
\newblock In {\em Proceedings of the 23rd ACM Symposium on Parallelism in
  Algorithms and Architectures (SPAA)}, pages 177--186, 2011.

\bibitem{Irani1994}
Sandy Irani.
\newblock Coloring inductive graphs on-line.
\newblock {\em Algorithmica}, 11(1):53--72, 1994.

\bibitem{Kesselheim2011}
Thomas Kesselheim.
\newblock A constant-factor approximation for wireless capacity maximization
  with power control in the {SINR} model.
\newblock In {\em Proceedings of the ACM-SIAM Symposium on Discrete Algorithms
  (SODA)}, pages 1549--1559, 2011.

\bibitem{Kesselheim2013ESA}
Thomas Kesselheim, Klaus Radke, Andreas T\"onnis, and Berthold V\"ocking.
\newblock An optimal online algorithm for weighted bipartite matching and
  extensions to combinatorial auctions.
\newblock In {\em Proceedings of the 21st annual European Symposium on
  Algorithms (ESA)}, pages 589--600, 2013.
\newblock To appear.

\bibitem{Kleinberg2012}
Robert Kleinberg and S.~Matthew Weinberg.
\newblock Matroid prophet inequalities.
\newblock In {\em Proceedings of the 44th annual ACM symposium on Theory of
  computing (STOC)}, pages 123--136, 2012.

\bibitem{Koren1995}
Gilad Koren and Dennis Shasha.
\newblock $\text{D}^{\textit{over}}$: An optimal on-line scheduling algorithm
  for overloaded uniprocessor real-time systems.
\newblock {\em SIAM Journal on Computing}, 24(2):318--339, 1995.

\bibitem{Korula2009}
Nitish Korula and Martin P{\'a}l.
\newblock Algorithms for secretary problems on graphs and hypergraphs.
\newblock In {\em Proceedings of the 36th International EATCS Colloquium on
  Automata, Languages and Programming (ICALP)}, pages 508--520, 2009.

\bibitem{Krengel1977}
Ulrich Krengel and Louis Sucheston.
\newblock Semiamarts and finite values.
\newblock {\em Bull. Amer. Math. Soc}, 83:745--747, 1977.

\bibitem{Krengel1978}
Ulrich Krengel and Louis Sucheston.
\newblock On semiamarts, amarts and processes with finite value.
\newblock {\em Advances in Prob}, 4:197--266, 1978.

\bibitem{KrystaV12}
Piotr Krysta and Berthold V{\"o}cking.
\newblock Online mechanism design (randomized rounding on the fly).
\newblock In {\em Proceedings of the 39th International EATCS Colloquium on
  Automata, Languages and Programming (ICALP)}, pages 636--647, 2012.

\bibitem{Molinaro2012}
Marco Molinaro and R.~Ravi.
\newblock Geometry of online packing linear programs.
\newblock In {\em Proceedings of the 39th International EATCS Colloquium on
  Automata, Languages and Programming (ICALP)}, pages 701--713, 2012.

\bibitem{Ye2009}
Yuli Ye and Allan Borodin.
\newblock Elimination graphs.
\newblock In {\em Proceedings of the 36th International EATCS Colloquium on
  Automata, Languages and Programming (ICALP)}, pages 774--785, 2009.

\end{thebibliography}
\bibliographystyle{plain}

\end{document}